\documentclass[journal]{IEEEtran}

\usepackage{graphics} 
\usepackage{epsfig} 
\usepackage{mathptmx} 
\usepackage{amsthm}
\usepackage{amsmath} 
\usepackage{amssymb}  
\usepackage{cuted}
\usepackage{float}
\usepackage{color}
\usepackage{url}
\usepackage{bm}

\usepackage[ruled,vlined,linesnumbered]{algorithm2e} 
\usepackage{tabularx}

\usepackage{subcaption,cite}
\usepackage[font=small]{caption}

\newtheorem{corollary}{Corollary}
\newtheorem{lemma}{Lemma}
\newtheorem{proposition}{Proposition}
\newtheorem{definition}{Definition}
\newtheorem{remark}{Remark}

\makeatletter 
\def\ps@IEEEtitlepagestyle{%
  \def\@oddfoot{\mycopyrightnotice}%
  \def\@evenfoot{}%
}
\def\mycopyrightnotice{%
  \begin{minipage}{\textwidth}
  \centering \scriptsize
  Copyright~\copyright~2025 IEEE. Personal use of this material is permitted. Permission from IEEE must be obtained for all other uses, in any current or future media, including reprinting/republishing this material for advertising or promotional purposes, creating new collective works, for resale or redistribution to servers or lists, or reuse of any copyrighted component of this work in other works. 
  \end{minipage}
}
\makeatother 

\begin{document}
%
\title{A Consolidated Game Framework for Cooperative Defense against Cross-Domain Cyber Attacks in Satellite-Enabled Internet of Things}
%
%
%
\author{Linan Huang, Peilong Liu, Xu Chen, Chunxiao Jiang,~\IEEEmembership{Fellow,~IEEE}, Linling Kuang,~\IEEEmembership{Member,~IEEE},  and~Jianhua~Lu,~\IEEEmembership{Fellow,~IEEE}
\thanks{
L. Huang, P. Liu, C. Jiang, and L. Kuang are with the Beijing National Research Center for Information Science and Technology (BNRist), Tsinghua University, Beijing 100084, China (e-mail:huanglinan@tsinghua.edu.cn; plliu@tsinghua.edu.cn; jchx@tsinghua.edu.cn; kll@tsinghua.edu.cn\}). 
}
\thanks{X. Chen and and J. Lu are with the Department of Electronic Engineering, Tsinghua University, Beijing 100084, China (e-mail: chenxu18@tsinghua.org.cn; lhh-dee@tsinghua.edu.cn).}
\thanks{L. Kuang and J. Lu are also with State Key Laboratory of Space Network and Communications.}
}

%
%

\markboth{Journal of \LaTeX\ Class Files,~Vol.~14, No.~8, August~2015}%
{Shell \MakeLowercase{\textit{et al.}}: Bare Demo of IEEEtran.cls for IEEE Journals}
%



\maketitle

\begin{abstract}
As the adoption of satellite-enabled Internet of Things (IoT) continues to rise, its intricate multi-domain architecture becomes increasingly susceptible to cross-domain cyber threats. 
Attackers can exploit compromised IoT devices, inject malicious packets into data streams aggregated at the IoT gateway for satellite backhaul, and potentially endanger the satellite network during transmission by exploiting the hardware, software, and protocol vulnerabilities.
Compared to single-domain defenses, cooperative defense at the IoT devices, IoT access network, and satellite transmission network provides fine-granularity defense against cross-domain intelligent attacks.
However, quantifying cross-domain impacts and tilting incentive misalignment among different participants remain significant challenges, making systematic cooperative defense development a complex task. 
To address this, we develop a tripartite security game framework to characterize the impacts of attacks and defense methods across both the terrestrial and satellite domains.
Leveraging this game model, we devise flow pricing to optimally motivate the IoT Network Operator (IoT-NO) to prevent malicious packet infiltration into the satellite domain.
Subsequently, we propose efficient learning algorithms enabling both the IoT-NO to ascertain their ideal flow sampling strategies and the Satellite Service Provider (SAT-SP) to determine optimal flow pricing. 
The simulation results corroborate the effectiveness of the consolidated game in counteracting cross-domain cyber attacks and facilitating cooperative defense between the IoT-NO and the SAT-SP with non-aligned incentives. 
\end{abstract}


\section{Introduction}

The number of Internet of Things (IoT) devices worldwide reaches $\$15.14$ billion in 2023 and is estimated to rise to $\$29.42$ billions by 2030 \cite{satistic}, exerting significant pressure on traditional terrestrial networks. 
{On one hand, the multitude of connections from IoT devices generates substantial data traffic, exacerbating network congestion in densely populated areas such as smart cities. On the other hand, terrestrial networks often struggle to meet data backhaul demands in remote areas, like those used for environmental monitoring, where cellular coverage is sparse and of poor quality. 
Consequently, the integration of satellite constellations for offloading and backhauling has garnered increasing interest in various IoT scenarios,}  owing to the ubiquitous coverage and high capacity of satellite networks \cite{soret20215g, centenaro2021survey}. 

{
The integration of satellite network and terrestrial IoT network constitutes a multi-domain structure, as illustrated in Fig. \ref{fig:multidomainStructure}.} 
An IoT Network Operator (IoT-NO) {manages} both the access network and the core network to provide clients with remote monitoring and control of their IoT devices. 
A Satellite Service Provider (SAT-SP) {deploys the} satellite network to transmit data flows from the access network to the core network. 
Such multi-domain integration leads to cross-domain cyberattacks, which originate in the IoT domain and impact the satellite domain. 
These attacks typically unfold in three distinct steps. 
\textcolor{black}{First, the hardware vulnerabilities identified at the Black Hat security conference \cite{wired2022starlink} demonstrate the feasibility of an attacker inexpensively compromising IoT devices.} 
Second, these compromised devices generate malicious data packets that legitimately pass through the IoT gateway and enter the satellite domain. 
Finally, when these malicious packets are processed by critical satellite components (e.g., the on-board computer and digital processing transponder for data processing and switching), adversaries can exploit software, hardware, and protocol vulnerabilities to steal confidential data, disrupt Command and Control (C2), and even take over satellite control \cite{SPARTA}. 


\begin{figure}[h]
\centering
\includegraphics[width=.9 \columnwidth]{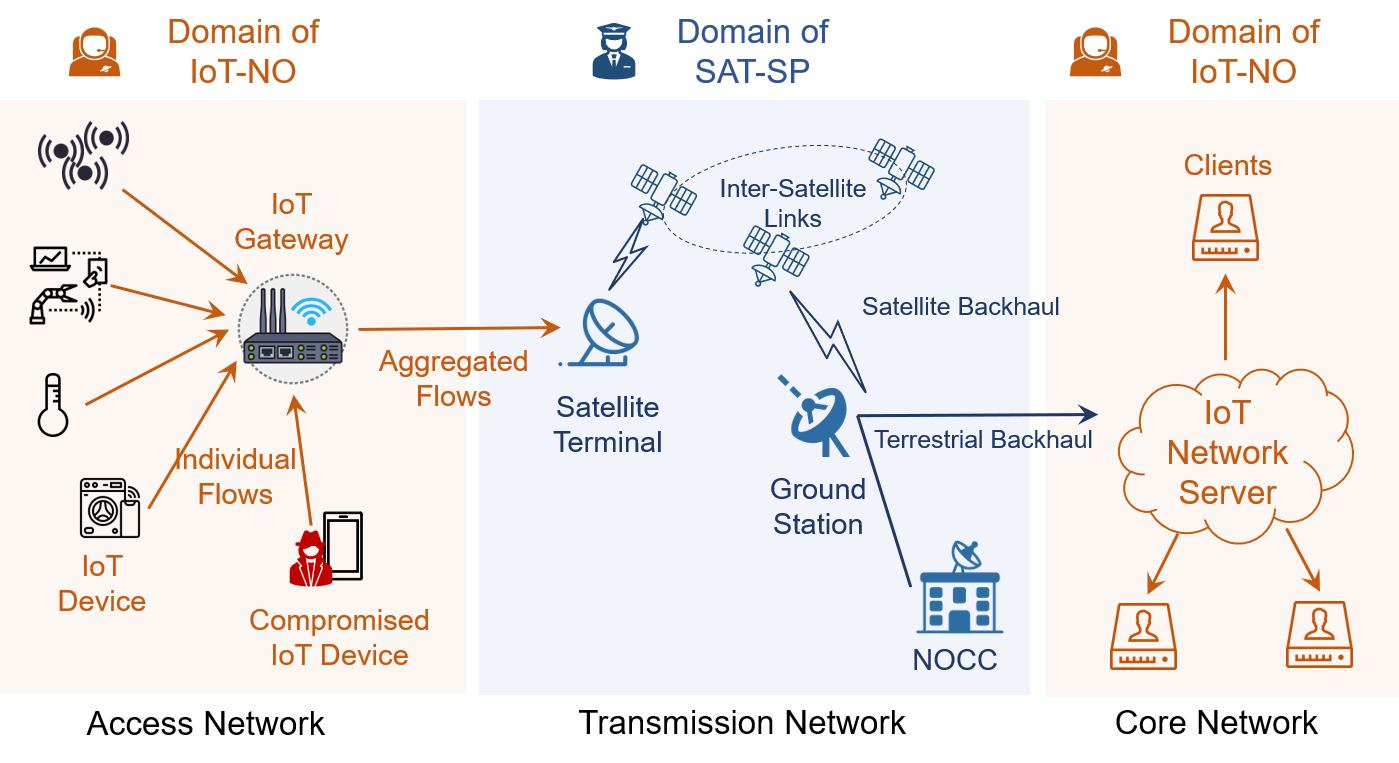}
\caption{ 
Multi-domain structure of satellite-enabled IoT {managed by} the IoT-NO and the SAT-SP in orange and blue, respectively. 
The IoT gateway aggregates the data flows from IoT devices and forwards the aggregated flow to a satellite terminal. The satellite network relays the aggregated flow to the ground station, which provides terrestrial backhaul to the IoT core network. 
The multi-domain integration leads to cross-domain cyberattacks that originate in the terrestrial IoT network while take effect in the satellite network. 
}
\label{fig:multidomainStructure}
\end{figure}

Although the ultimate attack targets are critical satellite components, single-domain defense mechanisms within the satellite network are insufficient under the multi-domain structure in Fig. \ref{fig:multidomainStructure}.
\textcolor{black}{To safeguard the privacy of IoT devices, the IoT-NO refrains from sharing their information with the SAT-SP and may even encrypt IoT device data until it reaches the IoT network server within the core network.}   
Lacking access to information about the IoT devices, the communication protocols (e.g., data frame formats and encryption schemes), or the diverse application-level protocols of the IoT network, the SAT-SP is unable to audit the data {transmitted through} the satellite network. 
Consequently, the SAT-SP can only deploy costly and coarse-granularity defenses (e.g., Dynamic Heterogeneous Redundancy (DHR) \cite{hu2024unveiling} and honeypots \cite{harrison2021defense,HoneypotSat}) that can detect the malicious data in aggregated data flows but cannot trace it back to the compromised IoT devices generating the data.  
Although disconnecting the aggregate flow can prevent malicious individual flows, it also blocks legitimate flows from IoT devices, leading to significant costs and inefficiencies, especially when massive IoT devices are connected to the satellite terminal. 




{
The key limitation of single-domain defense within satellite networks is the absence of a systematic vision that fails to exploit the attack's cross-domain feature. 
Since the attacks originate in the terrestrial IoT network and take effect in the satellite network, we can implement preventive defense methods at the IoT devices to mitigate the initial compromise and incorporate data inspection of individual flows within the access network for fine-grained defense. Such a multi-domain cooperative defense
effectively terminates the cross-domain cyber kill chain but also presents several challenges. 
}
First, the defense techniques within each domain have interdependent and cross-domain impacts. Consequently, finding an optimal way to integrate these multi-domain defense mechanisms to secure the entire satellite-enabled IoT system, while minimizing costs, presents a significant challenge. 
Second, IoT and satellite networks are managed by separate entities, each with its own objectives. While packet sampling and audit at the IoT gateway could lower the risk of malicious packets reaching the satellite domain, such measures can escalate operational expenses for the IoT-NO. Consequently, the IoT-NO may lack the incentive to apply the defense. 
Third, given the absence of precise system modeling, each participant needs to efficiently learn the optimal security strategy by striking a balance between exploration and exploitation. Furthermore, the security actions of other players may only be indirectly discerned through their resultant effects. 

To address these challenges, we propose a tripartite cross-domain security game involving the attacker, the SAT-SP, and the IoT-NO, covering both terrestrial and satellite domains. The game model provides a holistic characterization of the cross-domain impacts of the attack and defense actions. 
Using the game model as a stepping stone, we further investigate a flow pricing problem to incentivize the IoT-NO to actively engage in cooperative defense.
Finally, we formulate the learning problems of the IoT-NO and the SAT-SP into a multi-arm bandit problem and a best-response action learning problem, respectively.
Building on the theoretical underpinning that the price change has a piecewise-linear impact on the utilities of the IoT-NO and the SAT-SP, we develop an efficient finite-step learning algorithm. 
The numerical results from a case study corroborate the insufficiency of single-domain defense and the effectiveness of our consolidated game framework for cross-domain attacks. Moreover, the results lead to security insights, including the losses and mitigation of the Profit of Scale (PoS), gambling effects, and security-efficiency tradeoffs. 

\subsection{Notations and Organization of the Paper}
Calligraphic letter $\mathcal{Y}$ defines a set, and $\Delta \mathcal{Y}$ represents the set of probability distributions over set $\mathcal{Y}$. 
Subscripts and superscripts represent elements and different categories, respectively. 
{Let $\mu_{1:N}^{Mal}$ denote the tuple $(\mu_{1}^{Mal}, \cdots, \mu_{N}^{Mal})$ as a shorthand notation.}  
The rest of the paper is organized as follows. 
{We present the related works and the game model in Sections \ref{sec:related works} and \ref{sec:Game-Theoretic Model}, respectively.}   
Based on the analysis {of the cross-domain impact} in Section \ref{sec:equilibirum computation and analysis}, we formulate the price design problem and provide learning algorithms in Section \ref{sec:Detailed-Price Design and Learning}. 
Finally, Section \ref{sec:case study} evaluates the performance, and Section \ref{sec:conclusion} concludes the paper. 

\section{Related Works}
\label{sec:related works}
\subsection{\textcolor{black}{Cybersecurity Frameworks for Satellite Systems and Multidomain Collaborative Defense}}
\textcolor{black}{Traditional satellite systems, developed during the Cold War, relied on limited satellites and ground stations for point-to-point communication, facing threats like eavesdropping, jamming, and forgery \cite{guo2021survey}, addressed by cryptography and physical layer security \cite{zhang2023survey}. As they transition from specialized ``old space" systems to the digitized, internet-connected ``new space," they encounter new cybersecurity threats \cite{manulis2021cyber}.}

{
As a cutting-edge research field, cybersecurity for satellite networks has received widespread attention in recent years. 
From an attack perspective, the Center for Strategic and International Studies (CSIS)  emphasizes the threat of cyber weapons to satellite systems and space assets in its fifth edition of the Space Threat Assessment report \cite{Harrison_Johnson_Young_Wood_Goessler_2023}. 
In the same year of 2022, the Aerospace Corporation releases the Space Attack Research and Tactics Analysis (SPARTA) framework  \cite{SPARTA}, describing the unique threats that attackers may pose to space systems in areas such as intelligence gathering, initial intrusion, lateral movement, and defense evasion.}
From a defense perspective, the National Institute of Standards and Technology (NIST) has released technical reports addressing satellite network security \textcolor{black}{across ground, user, and space segments, including IR-8401 \cite{NIST-8401} for ground segment protection, IR-8270 \cite{NIST-8270} for identifying cybersecurity risks to space systems, and IR-8441 \cite{NIST-8441} for Hybrid Satellite Networks (HSNs).}   
The aforementioned reports, databases, and frameworks provide a panoramic view of satellite network security by illustrating  the Tactics, Techniques, and Procedures (TTP) of attacks and offering guidance for potential defense methods. 
However, they leave gaps in modeling strategic attack behaviors, quantifying the security impacts, and designing cost-effective defense strategies, which our work aims to fill. 

\textcolor{black}{Satellite-IoT falls under the broader category of Space-Air-Ground Integrated Network (SAGIN) \cite{liu2018space}. 
The multi-domain structure of SAGIN is recognized for improving overall network efficiency, enabling advancements such as resource orchestration \cite{zhang2021space}, dynamic traffic offloading \cite{tang2021deep}, and controller placement \cite{chen2022hierarchical}.} However, its implications for security and the development of cooperative defensive strategies remain underexplored, largely due to the complexities in capturing cross-domain impacts of strategic attacks. 
Prior works, such as \cite{guo2022distributed} and \cite{chen2022real}, adopt a multi-domain collaboration perspective to enable early detection and mitigate DDoS attacks at their initial entry points, thereby minimizing their impact. 
\textcolor{black}{Rather than primarily relying on blockchain for collaborative security measures}, our methodology, rooted in a game-theoretical approach, delves deeper into understanding and designing the incentives for strategic and intelligent participants.

\subsection{\textcolor{black}{Game Theory and Incentive Mechanisms in the Context of Cross-Domain Cyber Attacks}}
\textcolor{black}{
Game theory has been applied to cybersecurity since the early 2000s, valued for its ability to model and quantify strategic interactions between intelligent attackers and defenders \cite{manshaei2013game}. As attacks have evolved from isolated, one-shot events to persistent, multi-domain, and multi-layered threats, various game-theoretic models have been developed to address these complexities. From the perspective of participants and their incentives, these developments can be broadly categorized into three categories.} 

\textcolor{black}{
In the first category, researchers focus primarily on the confrontation between a single attacker and a single defender, often extending the analysis to multiple stages.  For instance, \cite{huang2019game} applies stochastic game theory, incorporating quantitative vulnerability analysis and time-based unified payoff quantification, to defend against attacks that span both cyber and physical layers. 
The authors in \cite{zhao2020finite} further explicitly incorporate time into the model using a finite-horizon Semi-Markov game to address time-sensitive cyber-physical attacks. 
To account for scenarios in which the defender has incomplete information about the attacker's payoff matrix, studies such as \cite{yao2024bayesian} propose a Bayesian-Stochastic hybrid game-theoretic framework for devising optimal cross-layer defensive strategies. 
These works commonly assume a persistent adversary capable of launching multi-stage attacks in terrestrial cyber-physical systems. However, in satellite IoT scenarios, attackers typically can only compromise IoT devices and initiate attacks at entry points, potentially propagating malicious impacts into the satellite domain.} 

\textcolor{black}{In the second category, researchers refine the concept of a ``defender" into multiple categories of participants, including legitimate users and diverse types of defenders, and have explored (potentially multi-domain) cooperative defense strategies. For instance, \cite{huang2020dynamic} proposes a dynamic Bayesian game to counter stealthy attackers who remain indistinguishable from legitimate users during cross-domain interactions. Other works have examined multiple defenders, such as Intrusion Detection Systems (IDS) \cite{liu2017energy} and regions within supervisory control and data acquisition (SCADA) systems, which interact to devise optimal cooperative defense strategies. 
Despite these advancements, existing approaches often treat multiple defensive participants uniformly, representing their interactions within the payoff matrix of a single game. 
In contrast, our work proposes a consolidated game model that captures the heterogeneous interactions specific to the IoT and satellite domains. This approach introduces a novel equilibrium concept, enabling the prediction of interdependent behaviors and cross-domain impacts among attackers, the IoT-NO, and the SAT-SP.}

\textcolor{black}{In the third category, with the incorporation of multiple defensive participants, the focus shifts from player interactions to internal and external incentive mechanisms.
Internal incentive mechanisms primarily address the tradeoff between Quality of Service (QoS) and security for individual players. For instance, \cite{sun2021cross} models the conflicting objectives of QoS and security within a vehicle as two abstract players competing for limited resources. Similarly, \cite{fadlullah2016gt} adopts a game-theoretic approach to jointly optimize QoS and security, enabling UEs to balance QoS and security levels while allowing eNBs to maximize bandwidth utilization.
External incentive mechanisms, on the other hand, conceptualize security as a service, encompassing encryption \cite{cai2022security}, security planning \cite{feng2020joint}, and computation offloading \cite{xu2018designing}. 
When security-related services are provided to multiple users, their interdependencies are further captured in \cite{feng2020joint,xu2018designing}. Moving beyond traditional monetary-based incentives, recent works such as \cite{huang2021duplicity,huang2023zetar} offer greater flexibility by leveraging information as an incentive to engage rational players. 
In these works, incentive mechanisms are typically designed by an independent coordinator, with metrics like overhead (e.g., the size increment of IPsec over IP in \cite{fadlullah2016gt}) and the strength of encryption mechanisms (e.g., key length in \cite{cai2022security}) used to quantify the security impact.
In our scenarios, the SAT-SP acts as both the flow price designer and a participant, directly engaging in game interactions. Furthermore, we explicitly incorporate the attacker’s behavior in selecting attack intensity, enabling a more comprehensive characterization of the cross-domain security impact under strategic attacks.}

\section{Consolidated Game Model}
\label{sec:Game-Theoretic Model}

{As illustrated in Fig. \ref{fig:CrossLayerGame}, the consolidated game framework offers a \textit{defense-in-depth} approach to the security of satellite-enabled IoT through collaborative defense at the device level, the individual flow level, and the aggregated flow level. 
At the device level, preventive security methods such as periodic system reinstallation are applied locally at each device to counter potential compromises. 
After compromising an IoT device, the attacker can inject malicious data in the device's individual data flow. 
At the individual flow level, the IoT-NO samples each device's individual flow at the IoT gateway for Deep Packet Inspection (DPI). Once a malicious packet is detected in a device's data flow, its access is revoked. 
At the aggregated flow level, the SAT-SP employs coarse-granularity defenses such as honeypots. These can assess the legitimacy of the aggregated data flow, but not of the individual flows. 
Once the attack payload in the aggregated flow from a satellite terminal triggers the honeypot, the SAT-SP terminates the satellite terminal's access.} 

{These collaborative defenses have an interdependent impact on terrestrial and satellite domains.} 
As highlighted by the orange color in Fig. \ref{fig:CrossLayerGame}, the terrestrial domain interaction involves the attacker, who determines the Malicious Packet Generation Rate (MPGR), and the IoT-NO, who determines the Sampling Frequency (SF) of all individual flows. 
As highlighted by the blue color in Fig. \ref{fig:CrossLayerGame}, the satellite domain interaction involves the IoT-NO and the SAT-SP, where the latter determines the security configuration of the coarse-granularity
defense. 
{The MPGR and the SF together determine the outcome of the individual flow DPI, which further affects the expected flow rates of the malicious and the legitimate packets arriving at the satellite domain. 
These flow rates and the SAT-SP's security configuration affect the utilities of both the IoT-NO and the SAT-SP, subsequently influencing their equilibrium strategies. 
Such coupling between terrestrial and satellite domain interactions leads to a consolidated equilibrium concept, the existence of which we prove under a mild condition.}  

\begin{figure}[h]
\centering
\includegraphics[width=.9 \columnwidth]{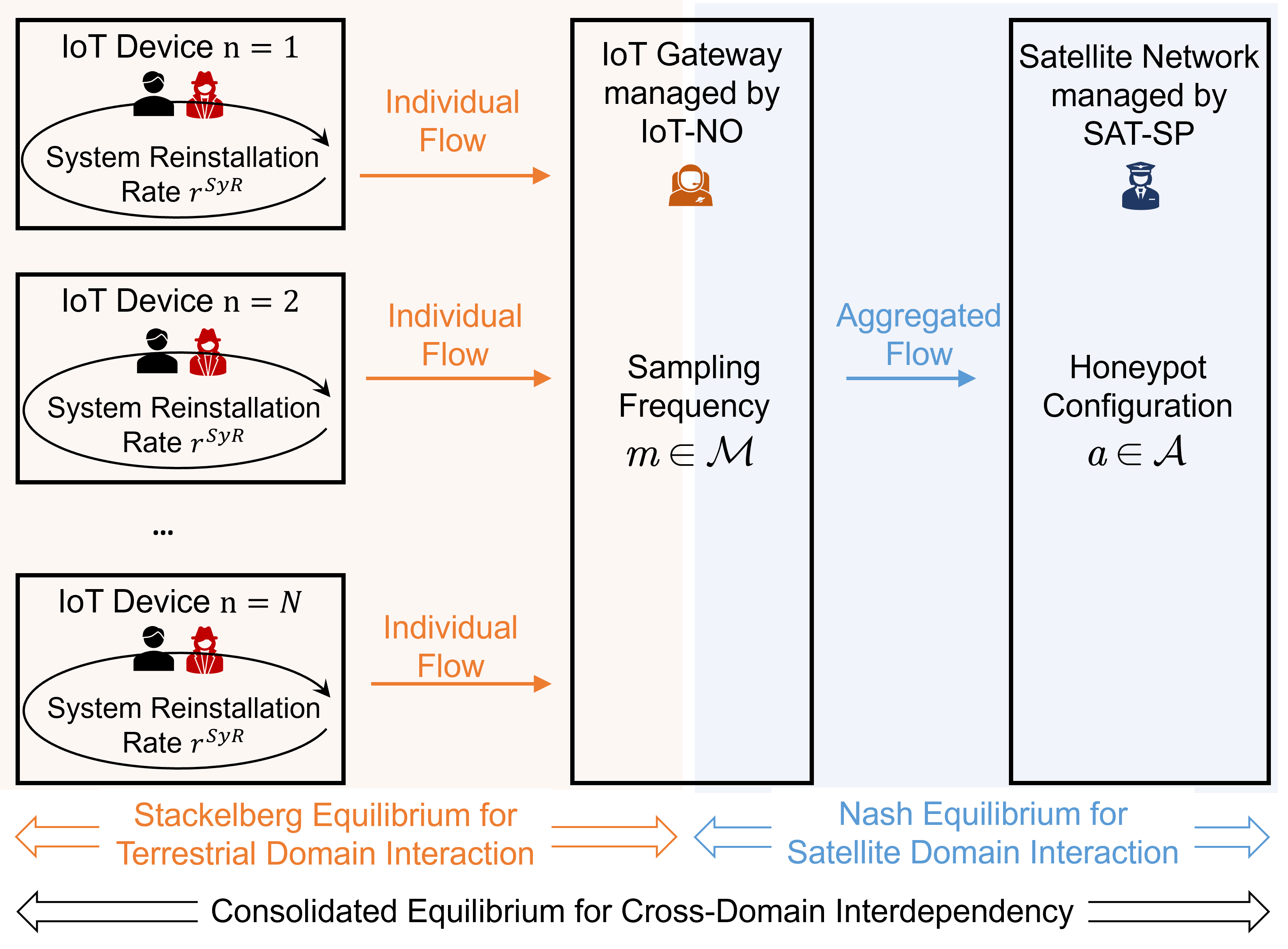}
\caption{ 
An overview of the consolidated game among the IoT-NO, the SAT-SP, and the attacker within the terrestrial and satellite domains. 
}
\label{fig:CrossLayerGame}
\end{figure}

Sections \ref{sec:Dynamic Compromise and Defense of IoT Devices}, \ref{sec:flow sampling detection}, and \ref{sec:Aggregated Flow and Honeypot Configuration} illustrate the collaborative defense at the device level, the individual flow level, and the aggregated flow level, respectively. 
{To elaborate on the coupling between terrestrial and satellite domain interactions,} we introduce the utility functions in Section \ref{sec:utilty}, followed by the {definition and computation} of the consolidated equilibrium concept in Section \ref{sec:Equilibrium Concept}. 

\subsection{Dynamic Compromise and Defense of IoT Devices}
\label{sec:Dynamic Compromise and Defense of IoT Devices}
As illustrated in Fig. \ref{fig:multidomainStructure}, a set of $N$ IoT devices connect to the IoT gateway and subsequently employ the satellite network for backhaul. 
\textcolor{black}{
These IoT devices, with restricted computation, memory, radio bandwidth, and battery resources, typically adopt lightweight security protections, making them more vulnerable to attacks than traditional computer systems \cite{xiao2018iot}. In addition to the hardware vulnerabilities \cite{wired2022starlink} discussed in the introduction, attackers can carry out a range of other attacks \cite{tsiknas2021cyber}, including supply chain attacks (where attackers can plant malicious components or code in IoT devices during manufacturing, ensuring persistent access), insecure update mechanisms (where malicious firmware updates can be injected by intercepting or spoofing legitimate updates, establishing control without detection), and weak authentication (where lack of proper authentication allows attackers to gain unauthorized access, inject malware, and maintain control over IoT devices).} 

\textcolor{black}{To model the initial compromise of the IoT devices, we define the state of} each device $n\in \mathcal{N}:=\{1,2,...,N\}$ as $x_n\in \mathcal{X}:=\{X^{Idl},X^{Nor},X^{Mal}\}$ where the transition probabilities among these states are illustrated in Fig. \ref{fig:StatreTrans}. 
\textcolor{black}{Throughout the paper, only the idle state $X^{Idl}$ is fully observable. 
It is challenging for the IoT-NO to detect whether the inner state of each IoT device is compromised (i.e., $X^{Mal}$) or normal (i.e., $X^{Nor}$). 
Instead of directly observing the internal state, the IoT-NO focuses on monitoring the observable outcomes, specifically whether individual flows contain malicious packets. Since it is the malicious packets, rather than the internal states, that affect the defense in the satellite domain, this approach aligns with the goal of mitigating cross-domain impacts.}

\begin{figure}[h]
\centering
\includegraphics[width=1 \columnwidth]{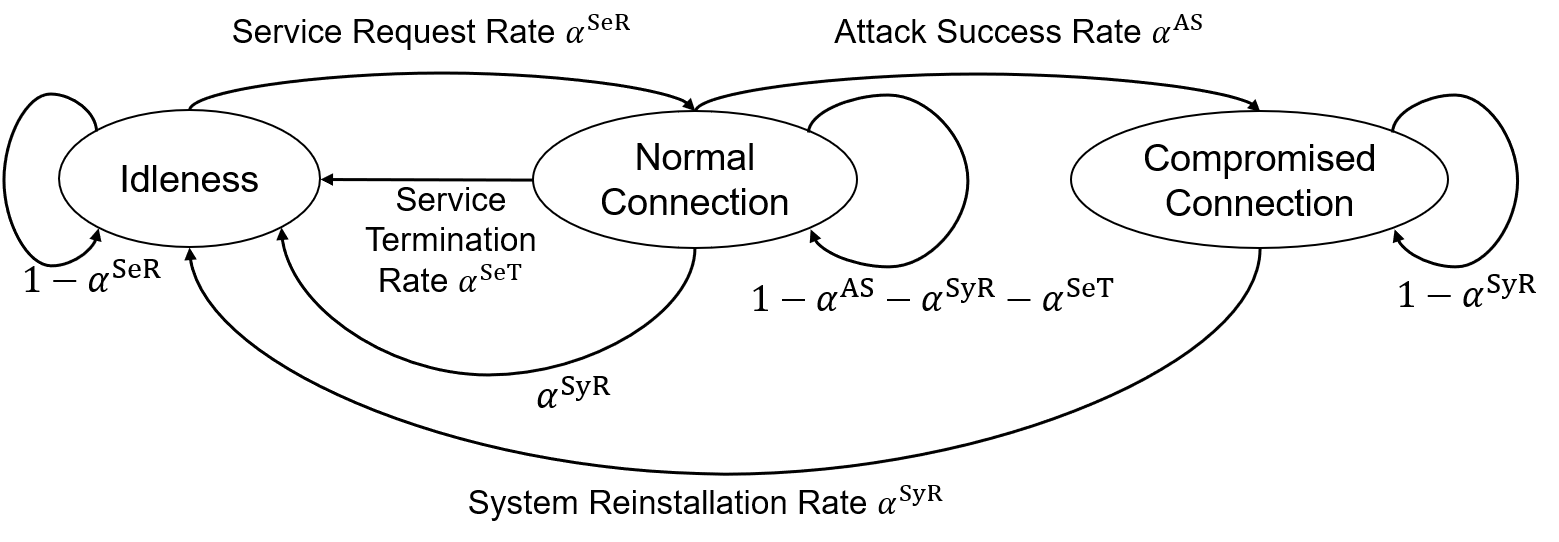}
\caption{ 
Transition probabilities among three states: idle $X^{Idl}$, normal connection $X^{Nor}$, and compromised connection $X^{Mal}$. 
The rates of service request, attack success, system reinstallation, and service termination are denoted as ${\alpha^{SeR}},{\alpha^{AT}},{\alpha^{SyR}},{\alpha^{SeT}}\in [0,1]$, respectively. 
}
\label{fig:StatreTrans}
\end{figure}

An IoT device, denoted as $n \in \mathcal{N}$, transitions from an idle state to an active connection upon receiving a service request, occurring at a rate of ${\alpha^{SeR}} \in [0,1]$ for power efficiency. Once the connection is established, the device may be compromised at a rate of ${\alpha^{AT}} \in [0,1]$, or it may return to the idle state upon service completion at a rate of ${\alpha^{SeT}} \in [0,1]$. 
\textcolor{black}{Since state $X^{Mal}$ is not observable}, to counteract potential compromises, each device enforces a \textcolor{black}{preventive} security policy that involves periodic system reinstallation at a rate of ${\alpha^{SyR}} \in [0,1]$. 
This reinstallation process resets the IoT device to its idle state, eradicating any remnants of prior attacks and requiring new attacks to begin anew.

\subsection{Individual Flow Generation and Sampling}
\label{sec:flow sampling detection}
{Device $n\in \mathcal{N}$ in non-idle states, i.e., $X^{Nor}$ and $X^{Mal}$, generates legitimate service data at an average rate of $\mu_n^{Nor}\in \mathbb{R}^{0+}$.} 
When a device $n\in \mathcal{N}$ is compromised, i.e., $x_n=X^{Mal}$, then the attacker could choose to insert malicious packets amidst the normal packets with a MPGR $\mu^{Mal}_n\in \mathbb{R}^{0+}$. 
At the IoT gateway, the IoT-NO samples each non-idle device's data flow with a frequency of $m\in \mathcal{M}:=\{m_1,m_2,\cdots,m_M\}$ times for DPI, where $m_j\in \mathbb{Z}^+$ and $M\in \mathbb{Z}^+$. 
Without loss of generality, we let $2\leq m_1\leq m_2\leq \cdots \leq m_M$. 
At each sampling stage, a malicious packet of a compromised device $n\in\mathcal{N}$ will be sampled with a probability of 
\begin{equation}
\label{eq:sample probability}
P^{Sam}_n(\mu_n^{Mal})=\mu^{Mal}_n/(\mu^{Nor}_n+\mu^{Mal}_n)\in [0,1]. 
\end{equation} 
Since the attacker has the freedom to choose the MPGR, {we denote} the malicious packet hit probability $P^{Sam}_n$ of device $n\in \mathcal{N}$ {as} a function of $\mu_n^{Mal}$. 

{Focusing on the attacks that exploit known vulnerabilities of critical components in the satellite network, we assume that the DPI is sufficiently effective to  identify the malicious packet once sampled. Then, once} a malicious packet has been sampled among the $m$ samples of the $n$-th device's data flow, the IoT-NO can detect the compromise of the device. 
Therefore, an device $n\in\mathcal{N}$ \textcolor{black}{under the state of compromised connection, i.e., $x_n=X^{Mal}$, will be detected with a probability 
\begin{equation}
\label{eq:detection probability}
   P^{Det}_n(\mu_n^{Mal},m)=1-[1-P^{Sam}_n(\mu_n^{Mal})]^{m} \in [0,1]. 
\end{equation}
}
Higher SF $m\in \mathcal{M}$ increases the detection probability, as shown in Lemma \ref{lemma:detection probability increase with SF}. Meanwhile, it introduces a higher overhead, defined as $c^{SF}: \mathcal{M}\mapsto \mathbb{R}^{-}$, to the IoT-NO, i.e., function $c^{SF}$ decreases concerning $m\in \mathcal{M}$. 

\begin{lemma}
\label{lemma:detection probability increase with SF}
The detection probability $P^{Det}_n(\mu_n^{Mal},m)$ of a compromised device $n\in \mathcal{N}$ in \eqref{eq:detection probability} increases with the SF $m\in \mathcal{M}$ for any MPGR $\mu_n^{Mal}\in \mathbb{R}^{0+}$. 
\end{lemma}
\begin{proof}
    The proof directly follows from \eqref{eq:sample probability} and \eqref{eq:detection probability}.  
\end{proof}

Since the cross-domain attack targets the satellite network, we define the utility function $U_n: \mathbb{R}\times\mathcal{M}\mapsto \mathbb{R}$ of the attacker who compromises device $n\in \mathcal{N}$ in \eqref{eq:Attacker Utility} as the expected malicious flow rate that enters the satellite network. 
\begin{equation}
\label{eq:Attacker Utility}
    U_n(\mu_n^{Mal},m) :=  [1-P^{Det}_n(\mu_n^{Mal},m)] \mu^{Mal}_n. 
\end{equation}







\subsection{Aggregated Flow and {Coarse-Grid Defense Configuration}}
\label{sec:Aggregated Flow and Honeypot Configuration}
 
Despite the individual flow sampling scheme in Section \ref{sec:flow sampling detection}, malicious packets from each device may still penetrate the satellite network. 
Depending on the IoT-NO's selection of the SF $m\in \mathcal{M}$ and each device's normal flow rate $\mu^{Nor}_{n}\in \mathbb{R}^+$, we define $w^{Nor}(\mu_{1:N}^{Mal},m)$ and $w^{Mal}(\mu_{1:N}^{Mal},m)$ as the expected rates of normal and malicious flows that enter the satellite network, respectively, where $\mu_{1:N}^{Mal}:=(\mu_{1}^{Mal},\cdots,\mu_{N}^{Mal})$.  

As illustrated in Fig. \ref{fig:multidomainStructure}, the IoT gateway has aggregated the individual flows of all $N$ IoT devices{, employing unknown protocols and encryption schemes,} before they reach the satellite network. 
{Equipped with knowledge of potential vulnerabilities that could be exploited within the satellite domain,} the SAT-SP can deploy {coarse-granularity defenses such as the DHR architecture \cite{hu2024unveiling} and honeypots to assess the legitimacy of} the aggregated data flow. 
{For example, the SAT-SP can deploy and operate heterogeneous processors (e.g., ARM, RISCV, MIPS) simultaneously in an onboard computer. 
Since it is challenging for an attack to exploit vulnerabilities in all these heterogeneous processors simultaneously, discrepancies in their computation outcomes could indicate a potential compromise. 
Considering the common requirement for Triple Modular Redundancy (TMR) in space systems to counter the harsh space environment, such an architecture would be well-suited. 
Similarly, the SAT-SP can also deploy honeypots as a deceptive and proactive method to detect the presence of malicious packets in the aggregated data flow.} 

{However, due to the lack of information regarding individual flows, these coarse-granularity defenses} cannot discern or attribute the malicious flow to the specific compromised device. 
Thus, once {the coarse-granularity defense has triggered an alarm (e.g., the aggregated flow triggers a honeypot or the heterogeneous components in the DHR  architecture return inconsistent results)} in the satellite network, the SAT-SP terminates the connection between the IoT gateway and the satellite terminal. 
It disables the flows of all $N$ IoT devices from entering the satellite network. 
If honeypots are not triggered {or the computation results of all DHR components are consistent}, then the aggregated flow enters the IoT core network using the satellite backhaul. 


{We abstract the different configurations of the selected coarse-granularity defense into a finite set, denoted as $\mathcal{A}:={\{a_1,a_2,\cdots,a_K\}}$. For example, a configuration can specify the number of heterogeneous components in the DHR architecture or the type of vulnerabilities that a honeypot is designed to mimic. 
The SAT-SP can choose from these $K$ configurations, where each configuration $a \in \mathcal{A}$ results in a different alarm triggering probability $P^{Cap}: (\mathbb{R}^{0+})^N \times \mathcal{M} \times \mathcal{A} \mapsto [0,1]$ and an associated operation cost $c^{OC}:\mathcal{A} \mapsto \mathbb{R}^{-}$. 
Besides the configuration $a\in \mathcal{A}$, the alarm triggering probability also depends on the attacker's MPGR tuple $\mu_{1:N}^{Mal}\in (\mathbb{R}^{0+})^N$ as well as the IoT-NO's SF $m\in\mathcal{M}$.}



\subsection{Utilities of the SAT-SP and the IoT-NO}
\label{sec:utilty}

Define $U^{SP}:\mathcal{A}\times \mathcal{M} \times (\mathbb{R}^+)^N \mapsto \mathbb{R}$ and $U^{NO}:\mathcal{A}\times \mathcal{M} \times (\mathbb{R}^+)^N \mapsto \mathbb{R}$ as the utility functions of the SAT-SP and the IoT-NO, respectively. 
Define $r^{Cap}\in\mathbb{R}^+$ as the reward of terminating per unit malicious packets. 
Define $r^{SP}\in\mathbb{R}^+$  (resp. $r^{NO}\in\mathbb{R}^+$) as the fee that the SAT-SP charges the IoT-NO (resp. the IoT-NO charges {the} clients) for transmitting per unit legitimate packets to the IoT core network. 

As shown in \eqref{eq:SAT-SP Utility}, the SAT-SP’s utility function comprises the expected gain {from detecting illegitimacy in the aggregated flow}, the reward for transmitting legitimate packets, and {the operation cost associated with the broad-granularity defense under configuration $a \in \mathcal{A}$}, i.e.,  
\begin{equation}
 \begin{split}
 \label{eq:SAT-SP Utility}
      &  U^{SP}(a,m,\mu_{1:N}^{Mal})  =   r^{Cap} {P^{Cap} (\mu_{1:N}^{Mal},m,a) } w^{Mal}(\mu_{1:N}^{Mal},m)   \\
       &  + r^{SP} (1- {P^{Cap}(\mu_{1:N}^{Mal},m,a)})  w^{Nor}(\mu_{1:N}^{Mal},m) + c^{OC}(a). \\ 
 \end{split}
\end{equation}

As shown in \eqref{eq:IOT-O Utility}, the IoT-NO's utility function includes the cost of applying SF $m$ and the expected gain for transmitting legitimate packets, i.e., 
\begin{equation}
\label{eq:IOT-O Utility}
\begin{split}
   & U^{NO} (a,m,\mu_{1:N}^{Mal})  
     =  c^{SF}(m) 
 \\ 
    &  + (r^{NO}-r^{SP}) (1- {P^{Cap}(\mu_{1:N}^{Mal},m,a)}) w^{Nor}(\mu_{1:N}^{Mal},m) . 
\end{split}
\end{equation}

{According to \eqref{eq:SAT-SP Utility} and \eqref{eq:IOT-O Utility}, the IoT-NO pays the SAT-SP and charges {the} clients based on the size of the data transmission. 
The IoT-NO needs to strike a balance between the cost of implementing individual flow DPI and the profit derived from the price difference in transmitting the data flow.} 

\subsection{Equilibrium Concept and {Computation Method}}
\label{sec:Equilibrium Concept}

As shown in Fig. \ref{fig:CrossLayerGame}, the consolidated game involves three players within both the terrestrial and satellite domains. 
Regarding their utility functions $U_n$, $U^{SP}$, and $U^{NO}$, the attacker who compromises device $n\in\mathcal{N}$ aims to optimize the malicious flow rate $\mu_n^{Mal}\in\mathbb{R}^{0+}$, the SAT-SP aims to optimize the configuration $a\in\mathcal{A}$, and the IoT-NO aims to optimize the SF $m\in\mathcal{M}$, respectively. 

This work considers the general case where the SAT-SP and the IoT-NO can take mixed strategies, denoted by $\sigma^{SP}\in \Delta\mathcal{A}$ and $\sigma^{NO}\in \Delta \mathcal{M}$, respectively.  
Therefore, $\sigma^{SP}(a)$ represents the probability of the SAT-SP adopting configuration $a\in \mathcal{A}$. 
Similarly, $\sigma^{NO}(m)$ represents the probability of the IOT-NO selecting SF $m\in\mathcal{M}$. 
Define $\bar{U}^{SP}:\Delta\mathcal{A}\times \Delta\mathcal{M} \times \mathbb{R}^+ \mapsto \mathbb{R}$ and $\bar{U}^{NO}:\Delta\mathcal{A}\times \Delta\mathcal{M} \times \mathbb{R}^+ \mapsto \mathbb{R}$ as the expected utility functions of the SAT-SP and the IoT-NO, respectively, under their mixed strategies.  
Then, we obtain 
\begin{equation*}
    \begin{split}
       \bar{U}^{SP}(\sigma^{SP},\sigma^{NO},\mu_{1:N}^{Mal}) = \mathbb{E}_{a\sim \sigma^{SP}, m\sim \sigma^{NO}}U^{SP}(a,m,\mu_{1:N}^{Mal}) 
       \\
       =\sum_{a\in \mathcal{A}}  \sum_{m\in \mathcal{M}}  \sigma^{SP}(a) \sigma^{NO}(m)  U^{SP}(a,m,\mu_{1:N}^{Mal}). 
    \end{split}
\end{equation*}

\begin{equation*}
    \begin{split}
       \bar{U}^{NO}(\sigma^{SP},\sigma^{NO},\mu_{1:N}^{Mal}) = \mathbb{E}_{a\sim \sigma^{SP}, m\sim \sigma^{NO}} U^{NO}(a,m,\mu_{1:N}^{Mal})
       \\
       =\sum_{a\in \mathcal{A}}  \sum_{m\in \mathcal{M}}  \sigma^{SP}(a)  \sigma^{NO}(m) U^{NO}(a,m,\mu_{1:N}^{Mal}).
    \end{split}
\end{equation*}

In the satellite domain, the IoT-NO and the SAT-SP can keep their mixed strategies of the SF selection and the configuration secret, respectively. 
The advanced attacker can obtain the SF chosen by the IoT-NO and the device's legitimate flow rate. 
Thus, {the process of flow generation and sampling in the terrestrial domain constitutes} a Stackelberg game, where the IoT-NO acts as the leader and the attacker as the follower.  
Hence, the attacker can limit {the} search space of the optimal MPGR, denoted as $\hat{\mu}_n^{Mal}(m)$, for each compromised device $n\in \mathcal{N}$ under SF $m\in \mathcal{M}$ to the pure strategies without loss of generality, i.e., 
\begin{equation}
\label{eq:definition of attack's BR}
    \hat{\mu}_n^{Mal}(m)\in \arg\max_{\mu_n^{Mal}\in\mathbb{R}^{0+}} U_n(\mu_n^{Mal},m).  
\end{equation}

For the aforementioned consolidated game, we define the mixed-strategy Consolidated Stackelberg-Nash Equilibrium (CSNE) in Definition \ref{def:Stackelberg-Nash Equlibirum} and provide the necessary and sufficient condition in Lemma \ref{lemma:existence} for the existence of the mixed-strategy CSNE with finite action sets. 
We define the shorthand notation  
$\bar{U}^{NO}(\sigma^{SP},\sigma^{NO},\mu_{1:N}^{Mal,*})  :=\sum_{a\in \mathcal{A}}  \sum_{m\in \mathcal{M}}  \sigma^{SP}(a)  \sigma^{NO}(m) U^{NO}(a,m,\hat{\mu}_{1:N}^{Mal}(m))$.

\begin{definition}[\textbf{Consolidated Stackelberg-Nash Equilibrium}]
\label{def:Stackelberg-Nash Equlibirum}
    The triple ($\mu_{1:N}^{Mal,*}, \sigma^{SP,*}\in \Delta \mathcal{A}, \sigma^{NO,*}\in \Delta\mathcal{M}$) constitute a Consolidated Stackelberg-Nash Equilibrium (CSNE) if 
    \eqref{eq:Nash Equlibirum} holds for all $\sigma^{SP}\in \Delta \mathcal{A}$ and $\sigma^{NO}\in \Delta \mathcal{M}$. 
        \begin{equation}
        \begin{split}
        \label{eq:Nash Equlibirum}
        &  \bar{U}^{SP}(\sigma^{SP,*},\sigma^{NO,*},\mu_{1:N}^{Mal,*}) \geq \bar{U}^{SP}(\sigma^{SP},\sigma^{NO,*},\mu_{1:N}^{Mal,*}). \\
        & \bar{U}^{NO} (\sigma^{SP,*},\sigma^{NO,*},\mu_{1:N}^{Mal,*})\geq \bar{U}^{NO} (\sigma^{SP,*},\sigma^{NO},\mu_{1:N}^{Mal,*}). 
        \end{split}
    \end{equation}
\end{definition}

\begin{lemma}[\textbf{The Existence of CSNE}]
\label{lemma:existence}
    Consider a consolidated game with finite action sets $\mathcal{M},\mathcal{A}$ and arbitrary utility functions of $U_n: \mathbb{R}\times\mathcal{M}\mapsto \mathbb{R}$, $U^{SP}:\mathcal{A}\times \mathcal{M} \times (\mathbb{R}^+)^N \mapsto \mathbb{R}$, and $U^{NO}:\mathcal{A}\times \mathcal{M} \times (\mathbb{R}^+)^N \mapsto \mathbb{R}$. 
    A mixed-strategy CSNE in Definition \ref{def:Stackelberg-Nash Equlibirum} exists if and only if $ \hat{\mu}_n^{Mal}(m)$ in \eqref{eq:definition of attack's BR} exists for all $m\in \mathcal{M}$. 
\end{lemma}

\begin{proof} 
    If the condition holds, then the consolidated game can be converted into a finite game with a guaranteed existence of an equilibrium in mixed strategy \cite{shoham2008multiagent}. 
    If the condition does not hold, then CSNE does not exist, as $\mu_{1:N}^{Mal,*}$ is undefined. 
\end{proof}


According to \cite{basar}, we can formulate the following bilinear program with linear constraints to compute the CSNE ($\mu_{1:N}^{Mal,*}\in \mathbb{R}^+, \sigma^{SP,*}\in \Delta \mathcal{A}, \sigma^{NO,*}\in \Delta\mathcal{M}$) in Definition \ref{def:Stackelberg-Nash Equlibirum}. 

\begin{proposition}[{\textbf{Program for CSNE Computation}}]
\label{proposition:bilinear program}
The triplet ($\mu_{1:N}^{Mal,*}, \sigma^{SP,*}\in \Delta \mathcal{A}, \sigma^{NO,*}\in \Delta\mathcal{M}$) constitutes a mixed-strategy CSNE if and only if \eqref{eq:Attacker's optimal malicious flow} holds, and $\sigma^{SP,*}\in \Delta \mathcal{A}, \sigma^{NO,*}\in \Delta\mathcal{M}$ is a solution to the following program: 
\begin{equation*}
\begin{split}
\label{eq:computation of CSNE}
 &   \max_{\sigma^{SP},\sigma^{NO},s^{SP},s^{NO}} s^{SP}+s^{NO} + \bar{U}^{SP}(\sigma^{SP},\sigma^{NO},\mu_{1:N}^{Mal,*}) \\
  & \quad\quad \quad \quad + \bar{U}^{NO}(\sigma^{SP},\sigma^{NO},\mu_{1:N}^{Mal,*}) \\
\textsc{s.t.} \quad
 & \sum_{m\in\mathcal{M}} \sigma^{NO}(m) U^{SP}(a,m,\hat{\mu}_{1:N}^{Mal}(m))  \leq - s^{SP} , \forall a\in\mathcal{A}, \\
 & \sum_{a\in\mathcal{A}} \sigma^{SP}(a) U^{NO}(a,m,\hat{\mu}_{1:N}^{Mal}(m))  \leq - s^{NO} , \forall m\in\mathcal{M}, \\
 & \ \sigma^{SP}\in \Delta \mathcal{A},\sigma^{NO}\in \Delta \mathcal{M}. 
\end{split}
\end{equation*}
The dimensions of the decision variables $\sigma^{SP}(a), a\in\mathcal{A}$ and $\sigma^{NO}(m), m\in\mathcal{M}$ are $|\mathcal{A}|$ and $|\mathcal{M}|$, respectively. The solution to the program exists and is achieved at the equality of constraints (a) and (b), i.e., $s^{SP,*}= - \bar{U}^{SP}(\sigma^{SP, *},\sigma^{NO, *},\mu_{1:N}^{Mal,*}) $ and $s^{NO, *}= - \bar{U}^{NO}(\sigma^{SP, *},\sigma^{NO, *},\mu_{1:N}^{Mal,*}) $. 
\end{proposition}


\section{{Cross-Domain Impact} Analysis}
\label{sec:equilibirum computation and analysis}
At the device level, we analyze the steady state of the IoT devices and the attacker's optimal MPGR in Sections \ref{sec:Steady State of IoT Devices} and \ref{sec:Attacker's Best-Response Strategy}, respectively. 
Based on the characterization of the expected rate of malicious and legitimate flows in Section \ref{sec:Expected Flow Rate in the Satellite Domain}, we provide theoretical underpinnings of the impact of the IoT-NO's SF on the SAT-SP's configuration selection in Section \ref{sec:Impact of SF on Configuration}. 

\subsection{Steady State of IoT Devices}
\label{sec:Steady State of IoT Devices}
The state transition in Fig. \ref{fig:StatreTrans} forms a Markov chain with the following transition matrix, i.e.,  
\begin{equation}
\label{eq:Tr}
Tr =     
    \begin{bmatrix}
    1-{\alpha^{SeR}} & {\alpha^{SyR}}+{\alpha^{SeT}} & {\alpha^{SyR}} \\
    {\alpha^{SeR}} & 1-{\alpha^{AT}}-{\alpha^{SyR}}-{\alpha^{SeT}} & 0 \\
    0 & {\alpha^{AT}} & 1-{\alpha^{SyR}} \\
    \end{bmatrix}. 
\end{equation}
{In \eqref{eq:Tr}, the columns represent the starting states, and the rows represent the ending states. The states are ordered as $X^{Idl},X^{Nor},X^{Mal}$. 
Therefore, $\sum_{i=1,2,3\}} Tr_{i,j}=1, \forall j\in \{1,2,3\}$.}  

{Driven by $Tr$ in \eqref{eq:Tr}, the states of the Markov chain reach a stationary distribution, denoted by $b=[b^{Idl}, b^{Nor}, b^{Mal}]'$, over the long term. Here,} $b^{Idl}, b^{Nor}, b^{Mal}\in [0,1]$ represent the probability of the $n$-th device's state being $X^{Idl},X^{Nor},X^{Mal}$, respectively, at the steady state. 
Define shorthand notation 
${\alpha^{tot}=\alpha^{SeR}} {\alpha^{AT}} + {\alpha^{SeR}} {\alpha^{SyR}} + {\alpha^{AT}} {\alpha^{SyR}} + ({\alpha^{SyR}})^2 + {\alpha^{SyR}} {\alpha^{SeT}}$. 
Since $b^{Idl}  +b^{Nor}  +b^{Mal}  =1$ and $b=Tr \cdot b$, we can solve $b$ as 
\begin{equation}
    \begin{split}
    \label{eq:closed form of b}
        b^{Idl}&= {\alpha^{SyR}} ({\alpha^{AT}} + {\alpha^{SyR}} + {\alpha^{SeT}}) / \alpha^{tot}. \\
        b^{Nor} &= {\alpha^{SeR}} {\alpha^{SyR}}/\alpha^{tot}. \\
        b^{Mal}  &= {\alpha^{SeR}} {\alpha^{AT}}/\alpha^{tot}.         
    \end{split}
\end{equation}

As shown in Lemma \ref{lemma:how alpha3 affects b}, the increases in the system reinstallation rate ${\alpha^{SyR}}$ can reduce the probability of the compromised state in the steady state. However, it can also hinder legitimate flow generation as the steady-state probability of the idle state increases. 
{We present the proof of Lemma \ref{lemma:how alpha3 affects b} in Appendix \ref{app:proof of Lemma 3,4,5}.}  

\begin{lemma}
\label{lemma:how alpha3 affects b}
    The values of $b^{Mal}$ and $b^{Idl}$ decrease and increase concerning ${\alpha^{SyR}}\in [0,1]$, respectively. 
    The value of $b^{Nor}$ increases and decreases  concerning ${\alpha^{SyR}}\in [0,\sqrt{{\alpha^{SeR}} {\alpha^{AT}}}]$ and ${\alpha^{SyR}}\in [\sqrt{{\alpha^{SeR}} {\alpha^{AT}}},1]$, respectively.  
\end{lemma}

The transition of a device state is on a much larger time scale than the flow generation and sampling processes in Section \ref{sec:flow sampling detection}. Thus, we consider the case where each device has reached its steady state and remained unchanged for the $m\in \mathcal{M}$ sampling stages. 

\subsection{Attacker's Best-Response Strategy}
\label{sec:Attacker's Best-Response Strategy}

Each device could be in the three states of $X^{Idl},X^{Nor},X^{Mal}$ with probabilities $b^{Idl}, b^{Nor}, b^{Mal}$. 
If a device $n\in\mathcal{N}$ is in states $X^{Idl}$ and $X^{Nor}$, then no flows and a normal flow with an average rate $\mu_n^{Nor}$ are generated, respectively.  
If a device $n\in\mathcal{N}$ is in the compromised state $X^{Mal}$, then the attacker can determine the MPGR $\mu_n^{Mal}$ to maximize {the} utility function for any given $m\in \mathcal{M}$. 
Lemma \ref{lemma:Attacker's optimal malicious flow} presents the closed-form representation of the attacker's best-response strategy $\hat{\mu}_n^{Mal}(m)$ defined in \eqref{eq:definition of attack's BR} and illustrates that a larger SF can help to reduce the MPGR. 
{We present the proof of Lemma \ref{lemma:Attacker's optimal malicious flow} in Appendix \ref{app:proof of Lemma 3,4,5}.} 
\begin{lemma}
\label{lemma:Attacker's optimal malicious flow}
The optimal MPGR of device $n\in \mathcal{N}$ under compromised state $X^{Mal}$ for a SF $m\in\mathcal{M}$ and legitimate flow rate $\mu^{Nor}_n$ has the following closed form representation, i.e., 
    \begin{equation}
    \label{eq:Attacker's optimal malicious flow}
    \hat{\mu}^{Mal}_n(m)= \mu^{Nor}_n / (m-1). 
\end{equation}
\end{lemma}
 

\subsection{Expected Flow Rate in the Satellite Domain}
\label{sec:Expected Flow Rate in the Satellite Domain}

Based on the attacker's optimal MPGR in \eqref{eq:Attacker's optimal malicious flow}, we can represent the detection probability in \eqref{eq:detection probability} as a function of $m\in\mathcal{M}$, as shown in \eqref{eq:pHat_Det}. 
\begin{equation}
\label{eq:pHat_Det}
   \hat{P}^{Det}(m):=P^{Det}_n(\hat{\mu}_n^{Mal}(m),m)= 1-(1-\frac{1}{m})^{m}, \forall n\in\mathcal{N}. 
\end{equation}

Define shorthand notations $h^{Mal}:=\sum_{n=1}^N b^{Mal}\mu_n^{Nor}$, $h^{Nor}:=\sum_{n=1}^N b^{Nor}\mu_n^{Nor}$, and {$\hat{P}^{Cap}(m,a):=P^{Cap}(\hat{\mu}_{1:N}^{Mal}(m),m,a)$}. 
Given the attacker's optimal MPGR in \eqref{eq:Attacker's optimal malicious flow}, we can compute the expected rate of normal and malicious flows that enter the satellite network in \eqref{eq:w_mal_compute} and \eqref{eq:w_Nor_compute}, respectively. 
Their monotonicities are shown in Lemma \ref{lemma:monotonicity} {with proof in Appendix \ref{app:proof of Lemma 3,4,5}.}   
\begin{equation}
    \begin{split}
    \label{eq:w_mal_compute}
 &  \hat{w}^{Mal}(m):=w^{Mal}(\hat{\mu}_{1:N}^{Mal}(m),m) 
 \\
& =\sum_{n=1}^N  b^{Mal} (1-\hat{P}_n^{Det}(m))\hat{\mu}_{n}^{Mal}(m) 
    = \frac{(m-1)^{m-1}}{m^m} h^{Mal}. 
    \end{split}
\end{equation}
\begin{equation}
    \begin{split}
    \label{eq:w_Nor_compute}
 & \hat{w}^{Nor}(m):= w^{Nor}(\hat{\mu}_{1:N}^{Mal}(m),m) 
  = \sum_{n=1}^N b^{Nor} \mu_{n}^{Nor} \\
&  + \sum_{n=1}^N  b^{Mal}  (1-\hat{P}_n^{Det}(m))\mu_{n}^{Nor}  = 
(1-\frac{1}{m})^m h^{Mal}+h^{Nor}. 
    \end{split}
\end{equation}

\begin{lemma}
\label{lemma:monotonicity}
    The functions $\hat{w}^{Mal}$, $\hat{P}^{Cap}$ and $\hat{w}^{Nor}$ decrease, decrease, and increase concerning SF $m\in \mathcal{M}$, respectively. 
\end{lemma}



\subsection{Impact of {the IoT's SF on the SAT-SP's Configuration}}
\label{sec:Impact of SF on Configuration}

Define the expected capture gain as 
\begin{equation}
    \begin{split}
    \label{eq:gm}
      &  g^{SP}(m):=r^{Cap}\hat{w}^{Mal}(m)-r^{SP}\hat{w}^{Nor}(m) \\
      &  = (r^{Cap}+r^{SP}-r^{SP}m)h^{Mal}(m-1)^{m-1}/m^m-r^{SP}h^{Nor}. 
    \end{split}
\end{equation}
Then we can represent the SAT-SP's utility function in \eqref{eq:SAT-SP Utility}  as 
$U^{SP}(a,m,\hat{\mu}_{1:N}^{Mal}(m))= g^{SP}(m)\hat{P}^{Cap}(m,a)+c^{OC}(a)+r^{SP}\hat{w}^{Nor}(m)$. 
Following Lemma \ref{lemma:monotonicity}, since both $\hat{w}^{Mal}(m)$ and $-\hat{w}^{Nor}(m)$ decrease with $m\in\mathcal{M}$, $g^{SP}(m)$ also decreases with $m\in\mathcal{M}$. 
We divide the SF set $\mathcal{M}$ into two disjoint subsets, denoted as the active region 
\begin{equation*}
    \mathcal{M}^{\leq}:=\{m\in \mathcal{M}| r^{Cap}\leq r^{SP}(m+ \frac{m^m h^{Nor}}{(m-1)^{m-1}h^{Mal}})\} 
\end{equation*}
and the passive region 
\begin{equation*}
    \mathcal{M}^{>}:=\{m\in \mathcal{M}|r^{Cap}>r^{SP}(m+ \frac{m^m h^{Nor}}{(m-1)^{m-1}h^{Mal}})\},  
\end{equation*}
where sets $\mathcal{M}^{\leq}$ and  $\mathcal{M}^{>}$ can be empty. 
Intuitively, from the SAT-SP's perspective, the IoT-NO has made adequate (resp. inadequate) security efforts for individual flow sampling and detection, {when SF $m\in \mathcal{M}^{\leq}$ (resp. $m\in \mathcal{M}^{>}$) is selected.}  

In Lemma \ref{lemma:Threshold Set-Division}, we characterize the above set division by a threshold $m^{TH}\in \mathcal{M}$ defined in Definition \ref{def:Set-Division Threshold}. 
As we will illustrate in Proposition \ref{proposition:two regions}, the threshold $m^{TH}\in \mathcal{M}$ measures the adequacy of the SF and exerts a threshold impact on the SAT-SP's best-response configuration. 

\begin{definition} [\textbf{Set-Division Threshold}]
\label{def:Set-Division Threshold}
Define $m^{TH}\in \mathcal{M}$ as the threshold to divide $\mathcal{M}$ into the active and passive regions. 
\begin{itemize}
    \item[(a).] If $m_M\in \mathcal{M}^{>}$, then $m^{TH}:=m_M$.
    \item[(b).] For $m_1<m<m_M$, if $m_j\in \mathcal{M}^{>}$ and $m_{j+1} \in \mathcal{M}^{\leq}$, then $m^{TH}:=m_j$.
    \item[(c).] If $m_1\in \mathcal{M}^{\leq}$, then $m^{TH}:=m_1$. 
\end{itemize}
\end{definition}   

\begin{lemma}[\textbf{Threshold Division}]
\label{lemma:Threshold Set-Division}
    The active and passive regions can be computed as $\mathcal{M}^{\leq}=\{m\in\mathcal{M}|m > m^{TH}\}$ and $\mathcal{M}^{>}=\{m\in\mathcal{M}|m \leq m^{TH}\}$, respectively. 
\end{lemma}
\begin{proof}
    Since the function $g^{SP}(m)$ in \eqref{eq:gm} decreases with $m\in \mathcal{M}$, it has at most one zero point. 
    Thus, for any $m,m'\in \mathcal{M}$, if $m'\in \mathcal{M}^{\leq}$ (i.e., $g^{SP}(m')\leq 0$), then $m > m'$ belongs to $\mathcal{M}^{\leq}$ (i.e., $g^{SP}(m)\leq g^{SP}(m')\leq 0$); if $m'\in \mathcal{M}^{>}$ (i.e., $g^{SP}(m')> 0$), then $m < m'$ belongs to $\mathcal{M}^{>}$ (i.e., $g^{SP}(m)> g^{SP}(m') > 0$).  
\end{proof}


{In the subsequent parts of this section, we specify $P^{Cap}(\mu_{1:N}^{Mal},m,a) = \rho(a) \bar{P}^{Cap}(\mu_{1:N}^{Mal},m)$, where $\bar{P}^{Cap}(\mu_{1:N}^{Mal},m)$ represents the baseline alarm triggering probability and $\rho:\mathcal{A} \mapsto [0,1]$ captures the efficiency of different configurations.}   
{We let the order of $a_1, \cdots, a_K$ reflects an increase in security effect or engagement level, where higher engagement levels are more effective but also more costly, i.e.,}   
$0\leq \rho(a_1)\leq \rho(a_2)\leq \cdots\leq \rho(a_K)\leq 1$ and $-c^{OC}(a_1)\leq -c^{OC}(a_2) \leq \cdots\leq -c^{OC}(a_K)$. 
Define $\hat{a}(m)\in \mathcal{A}$ as the best-response configuration of the SAT-SP concerning $m\in \mathcal{M}$, i.e., $\hat{a}(m)\in \arg\max_{a\in \mathcal{A}}U^{SP}(a,m,\hat{\mu}_{1:N}^{Mal}(m))$. 

\begin{proposition}[\textbf{Two Regions of Best-Response Configurations}]
\label{proposition:two regions}
    The best-response configurations under the active and passive regions are demonstrated in the following statements (a) and (b), respectively.  
    \begin{itemize}
        \item[(a).]  The least engagement level $a_1\in \mathcal{A}$ is a best-response configuration for the SAT-SP if $m\in \mathcal{M}^{\leq}$. 
        \item[(b).] For all $m\in \mathcal{M}^{>}$, then the SAT-SP chooses the configuration with the optimal tradeoff of operation cost and the capture gain, i.e., $\hat{a}^*(m)\in \arg\max_{a\in \mathcal{A}} \rho(a)g^{SP}(m)\hat{P}^{Cap}(m)+c^{OC}(a), \forall m\in \mathcal{M}^{>}$. 
    \end{itemize}
\end{proposition}
\begin{proof}
   To prove statement (a), for any given $m\in\mathcal{M}^{\leq}$, we obtain $r^{Cap}\leq r^{SP}(m+ \frac{m^mh^{Nor}}{(m-1)^{m-1}h^{Mal}})$, $g^{SP}(m)\leq 0$, and $U^{SP}(a_i,m,\hat{\mu}_{1:N}^{Mal}(m))\leq U^{SP}(a_j,m,\hat{\mu}_{1:N}^{Mal}(m)), \forall i\leq j, a_i,a_j\in\mathcal{A}$. 
   To prove statement (b),  if $r^{Cap}> r^{SP}(m+ \frac{m^mh^{Nor}}{(m-1)^{m-1}h^{Mal}})$, then $g^{SP}(m)$ is greater than $0$. 
\end{proof}

According to Lemma \ref{lemma:Threshold Set-Division} and Proposition \ref{proposition:two regions}, the SF set $\mathcal{M}$ can be divided into the active region $\mathcal{M}^{\leq}$ and the passive region  $\mathcal{M}^{>}$ with a threshold $m^{TH}\in \mathcal{M} \cup \emptyset$. 
The SF in the active region $\mathcal{M}^{\leq}$ samples each individual flow with adequate frequency (i.e., $m>m^{TH}$) such that the expected malicious flow is sufficiently low in the aggregated flow. 
In that case, the expected gain of flow transmission fee outweighs the expected reward to capture malicious flows, i.e.,  $g^{SP}(m)=r^{Cap}\hat{w}^{Mal}(m)-r^{SP}\hat{w}^{Nor}(m)\leq 0$, and the SAT-SP should apply the least engagement level due to the low risk. 

In contrast, the SF in the passive region $\mathcal{M}^{>}$ {is inadequate} (i.e., $m<m^{TH}$), resulting in a high amount of malicious data in the aggregated flow. 
To address the high cyber risk in the satellite network, the SAT-SP needs to increase the engagement level under the tradeoff of operation cost and capture gain. 
The statements (b) and (a) of Proposition \ref{proposition:two regions} lead to Corollaries \ref{corollary:1 and K} and \ref{corollary:Fundamental Limit of Security-Efficiency Ratio}, respectively.  
{Their proofs are presented in Appendix \ref{app:2corollary}.}  

\begin{corollary}
\label{corollary:1 and K}
    If $\tilde{m}\in \mathcal{M}^{>}$ and $\rho(\hat{a}^*(\tilde{m}))\geq \rho(a_i), \forall a_i\in \mathcal{A}$, then $\hat{a}^*(\tilde{m})\in \arg\max_{a\in \mathcal{A}}U^{SP}(a,m,\hat{\mu}_{1:N}^{Mal}(m)), m\in [m_1,\tilde{m}]$.  
\end{corollary}


Corollary \ref{corollary:1 and K} illustrates that in region two, where the SAT-SP has already applied the highest engagement level (i.e., $\rho(\hat{a}^*(\tilde{m}))\geq \rho(a_i), \forall a_i\in \mathcal{A}$), if IoT-NO further reduces security effort, then the  SAT-SP's best response is still to apply the highest engagement level. 


\begin{corollary}[\textbf{Fundamental Limit of Security-Efficiency Ratio for Dominant Strategy}]
\label{corollary:Fundamental Limit of Security-Efficiency Ratio}
    If the Security-Efficiency Ratio (SER) $r^{Cap}/r^{SP}$ is less than $2+4\alpha^{SyR}/{\alpha^{AT}}$, then the  configuration with the least engagement level $a_1\in \mathcal{A}$ is a dominant strategy of the SAT-SP, i.e., $U^{SP}(a_1,m,\hat{\mu}_{1:N}^{Ma}(m))\geq U^{SP}(a,m,\hat{\mu}_{1:N}^{Ma}(m)), \forall a\in \mathcal{A}, m\in \mathcal{M}$.   
\end{corollary}

Corollary \ref{corollary:Fundamental Limit of Security-Efficiency Ratio} states that if the ratio between the reward of preventing malicious flows and the gain of transiting normal flows is less than a threshold, then the SAT-SP should always configure honeypots with the least engagement level. 
In particular, the threshold of the SAT-SP's security decision in the satellite network is affected by the risk of all $N$ IoT devices, reflecting the cross-domain impact in Remark \ref{remark:cross-domain impact}. 
In particular, a higher Reinstallation-Attack Ratio (RAR) reduces the risk of IoT devices being compromised, resulting in less malicious data in the aggregated flow, and linearly increases the SER threshold to apply $a_1$ as the dominant strategy. 

\begin{remark}[\textbf{Cross-Domain Impact}]
\label{remark:cross-domain impact}
Corollary \ref{corollary:Fundamental Limit of Security-Efficiency Ratio} illustrates the cross-domain impact of the RAR ${\alpha^{SyR}}/{\alpha^{AT}}$ in the IoT device layer on the fundamental limit of the SER $r^{Cap}/r^{SP}$, i.e., the SER threshold to apply $a_1$ as the dominant strategy is proportional to the RAR. 
\end{remark}

\section{Price Design and Strategy Learning}
\label{sec:Detailed-Price Design and Learning}


Besides implementing  coarse-granularity  defense, the SAT-SP can further control the data transmission fee to incentivize the IoT-NO to increase the SF and facilitate cooperative defense across the terrestrial and satellite domains.
\textcolor{black}{
Compared with existing price-based incentive mechanisms (e.g., \cite{xu2018designing, feng2020joint}), our approach does not delve into the detailed design of the mechanisms. Instead, it emphasizes the system-level parameter $r^{SP}$, which directly influences the interactive behaviors between the IoT-NO and the SAT-SP. Since the focus of this paper is to develop a holistic, system-level design framework, abstracting $r^{SP}$ as the representative outcome of the incentive mechanism allows us to capture its key piecewise linear characteristics and leverage them to propose an efficient learning algorithm. 
}

To this end, we formulate the price design problem in Section \ref{sec:Price Design Problem Formulation}. Since the problem is NP-hard and the model parameters might not be precisely observed, we develop learning algorithms for both the IoT-NO and the SAT-SP to learn the optimal SF and the price in Sections \ref{sec:SF Learning Algorithm for the IoT-NO} and \ref{sec:Price Learning Algorithm for the SAT-SP}, respectively. 
The IoT-NO does not know the alarm triggering probability and aims to learn the optimal SF with the least regret. 
 The SAT-SP does not know the cost of the IoT-NO for applying different SFs, thus needs to adjust the price to learn its impact on the  IoT-NO's selection of SFs.
We first provide two participation conditions in Proposition \ref{proposition:IR for NO}.  
\begin{proposition} 
\label{proposition:IR for NO}
    If $(h^{Mal}+h^{Nor})(r^{NO}-r^{SP})\leq -c^{SF}(m_1)$, then $U^{NO}(a,m,\hat{\mu}^{Mal}_{1:N}(m))\leq 0, \forall m\in\mathcal{M},\forall a\in\mathcal{A}$. 
    If $(h^{Mal}+h^{Nor})r^{SP}+ h^{Mal} r^{Cap}/4 \leq -c^{OC}(a_1)$, then  $U^{SP}(a,m,\hat{\mu}^{Mal}_{1:N}(m))\leq 0, \forall m\in\mathcal{M},\forall a\in\mathcal{A}$. 
\end{proposition}
\begin{proof}
   According to \eqref{eq:IOT-O Utility} and  $h^{Mal}/4+h^{Nor}=w^{Nor}(2)\leq w^{Nor}(m) \leq  w^{Nor}(\infty)=h^{Mal} + h^{Nor}, \forall m\in \mathcal{M}$, if $r^{NO}<r^{SP}$, then $U^{NO}(a,m,\hat{\mu}^{Mal}_{1:N}(m))\leq c^{SF}(m)\leq 0$. 
   If $r^{NO}+c^{SF}(m_1)/(h^{Mal}+h^{Nor}) \leq r^{SP}\leq r^{NO}$, then $U^{NO}(a,m,\hat{\mu}^{Mal}_{1:N}(m))\leq c^{SF}(m)+(h^{Mal}+h^{Nor})(r^{NO}-r^{SP})\leq 0, \forall m\in\mathcal{M},\forall a\in\mathcal{A}$.  
   Combining the above two conditions, $r^{NO}+c^{SF}(m_1)/(h^{Mal}+h^{Nor}) \leq r^{SP}$ is a sufficient condition for $U^{NO}(a,m,\hat{\mu}^{Mal}_{1:N}(m))\leq 0, \forall m\in\mathcal{M},\forall a\in\mathcal{A}$. 
   Similarly, according to \eqref{eq:SAT-SP Utility} and $0=w^{Mal}(\infty)\leq w^{Mal}(m) \leq w^{Mal}(2)=h^{Mal}/4, \forall m\in \mathcal{M}$, we obtain the sufficient condition for a negative utility of the SAT-SP. 
\end{proof}
Proposition \ref{proposition:IR for NO} shows that if the IoT-NO's price gain under any configuration is less than the cost of applying the least-effort SF, then the IoT-NO has no incentive to participate. 
Similarly, if the sum of the service gain and the security gain is less than the cost of applying the configuration of the least engagement level, then SAT-SP has no incentive to run the backhaul service. 
Note that $\frac{-c^{OC}(a_1)-h^{Mal}r^{Cap}/4}{h^{Mal}+h^{Nor}}=-\frac{c^{OC}(a_1)}{h^{Mal}+h^{Nor}}-\frac{r^{Cap}{\alpha^{AT}}}{4({\alpha^{AT}}+{\alpha^{SyR}})}$. 
 Proposition \ref{proposition:IR for NO} provides a loose bound for the design of the price $r^{SP}$, as illustrated in Remark \ref{remark:IR}. 
 
\begin{remark}[\textbf{Participation Constraints}]
\label{remark:IR}
    According to Proposition \ref{proposition:IR for NO}, the SAT-SP's price design needs to at least satisfy $-\frac{c^{OC}(a_1)}{h^{Mal}+h^{Nor}}-\frac{r^{Cap}{\alpha^{AT}}}{4({\alpha^{AT}}+{\alpha^{SyR}})} \leq r^{SP}\leq r^{NO}+\frac{c^{SF}(m_1)}{h^{Mal}+h^{Nor}}$ to make the participation incentives of both the SAT-SP and the IoT-NO.  
\end{remark}

\subsection{Price Design Problem Formulation}
\label{sec:Price Design Problem Formulation}
\textcolor{black}{
The price design problem involves a trade-off from both the IoT-NO and the SAT-SP perspectives. On one hand, a lower flow price $r^{SP}$ boosts the IoT-NO’s profit from transmitting flows into the satellite domain, incentivizing it to put more effort into detecting and blocking malicious packets, which in turn benefits the SAT-SP by reducing malicious traffic. On the other hand, a lower price decreases the SAT-SP’s profit per unit of data transmission, necessitating a careful balance between these competing objectives.
}

In practice, the SAT-SP may not apply a mixed strategy for the configuration. 
Thus, for the remaining part of this section, we consider a predefined configuration $a_0\in\mathcal{A}$ and focus on designing the flow transmission fee $r^{SP}$. 
Then, the optimal SF for the IoT-NO can be achieved equivalently by a pure strategy under any given $r^{SP}$, and we represent both players' utilities as $\tilde{U}^{SP}(m,r^{SP}):=U^{SP}(a_0,m, \hat{\mu}^{Mal}_{1:N}(m))$ and $\tilde{U}^{NO}(m,r^{SP}):=U^{NO}(a_0,m, \hat{\mu}^{Mal}_{1:N}(m))$. 
The price design problem for the SAT-SP is formulated below. 
\begin{equation}
    \begin{split}
    \label{eq:priceDesign}
         \max_{r^{SP}\in \mathbb{R}^+,m^*\in\mathcal{M}} \  & \tilde{U}^{SP}(m^*,r^{SP})\\
         \textsc{s.t.} \quad\quad  & \tilde{U}^{NO}(m^*,r^{SP})\geq \tilde{U}^{NO}(m,r^{SP}), \forall  m\in\mathcal{M}, \\
                & \tilde{U}^{i}(m^*,r^{SP})\geq 0, i\in \{NO,SP\}.
    \end{split}
\end{equation}
The above program cannot be solved directly, as the optimization variable $m$ serves as the index. 
We can, however, try out all $M$ possible values of $m^*$ and pick the optimal one. 
Once $m^*$ is given, then it is a Linear Program (LP). So it means that we need to solve $M$ LPs solvable in polynomial time to find the solution. 
Since the problem is NP-hard, as shown in Proposition \ref{proposition:NP-Hard}, it is unlikely to find a more efficient algorithm than the above exhaustive search. 

\begin{proposition}
\label{proposition:NP-Hard}
    Program \eqref{eq:priceDesign} is equivalent to the Mixed-Integer Bilinear Programming (MIBP) in \eqref{eq:MIBP} and is NP-hard.  
    \begin{equation}
    \begin{split}
    \label{eq:MIBP}
    & \max_{r^{SP}\in \mathbb{R}^+,z_1,\cdots,z_M\in \{0,1\}, \sum_{m \in \mathcal{M}} z_m = 1 } \quad\quad \sum_{m \in \mathcal{M}} \tilde{U}^{SP}(m,r^{SP}) \cdot z_m \\
    & \textsc{s.t.}  \   \sum_{m \in \mathcal{M}} \tilde{U}^{NO}(m,r^{SP}) \cdot z_m \geq \tilde{U}^{NO}(m',r^{SP}), \forall  m' \in \mathcal{M}, \\
    & \quad\quad \sum_{m \in \mathcal{M}} \tilde{U}^{i}(m,r^{SP}) \cdot z_m \geq 0, i\in \{NO,SP\}. 
    \end{split}
\end{equation}
\end{proposition}
\begin{proof}
    The constraint $\sum_{m \in \mathcal{M}} z_m = 1$ on 
    binary decision variables $z_1,\cdots,z_M\in \{0,1\}$ guarantee that one and only one $z_m$ equals $1$ while all others equal $0$. Therefore, these binary decision variables at the optimal solution of \eqref{eq:MIBP} take the form of $z_m=\mathbf{1}_{\{m=m^*\}}$, which is equivalent to the original price design problem   \eqref{eq:priceDesign}. 
    Since MIBP is NP-hard, we know that \eqref{eq:priceDesign} is also NP-hard. 
\end{proof}

\subsection{SF Learning Algorithm for the IoT-NO}
\label{sec:SF Learning Algorithm for the IoT-NO}

The IoT-NO aims to select the optimal SF $m\in \mathcal{M}$ for a known parameter sets of $r^{SP}$, $c^{SF}$, and $r^{NO}$. 
{Not knowing the SAT-SP's configuration, the IoT-NO} 
needs to estimate {$\hat{P}^{Cap}(m,a_0)$} and $\hat{w}^{Nor}_{1:N}(m)$ by trying different SFs. 


After applying SF $m\in \mathcal{M}$, the IoT-NO can observe the following outcomes of the experiment: 
\begin{itemize}
    \item Whether {the alarms} in the satellite network are triggered (denoted as $y_0=1$) or not (denoted as $y_0=0$). 
    \item Whether each individual flow $n\in \mathcal{N}$ has been detected to contain malicious packets (denoted as $y_n=1$) or not (denoted as $y_n=0$). 
    \item The flow rate of each IoT device $n\in \mathcal{N}$, denoted as $\tilde{\mu}_n$. 
\end{itemize}

Since the observed flow rate $\tilde{\mu}_n$ is a combination of the legitimate and malicious flow rates, we need to estimate their percentages, as shown below. 

If a malicious packet is detected (i.e., $y_n=1$), then  $\tilde{\mu}_n=0$. 
\textcolor{black}{If $y_n=0$ and the IoT is not idle $\tilde{\mu}_n>0$, then there are two possibilities. 
First, device $n$ is at state $X^{Mal}$ with probability $\Pr(x_n=X^{Nor}|y_n=0,\tilde{\mu}_n>0;m)=\frac{ \Pr(y_n=0,\tilde{\mu}_n>0 | x_n=X^{Nor};m) b^{Nor} }{\Pr(y_n=0,\tilde{\mu}_n>0 |x_n=X^{Nor};m)b^{Nor}+\Pr(y_n=0,\tilde{\mu}_n>0 | x_n=X^{Mal};m)b^{Mal}}$, where $\Pr(y_n=0,\tilde{\mu}_n>0 | x_n=X^{Mal};m)=1-\hat{P}^{Det}(m)$ and $\Pr(y_n=0,\tilde{\mu}_n>0 | x_n=X^{Nor};m)=1$ for all $m\in \mathcal{M}$. 
If $x_n=X^{Mal}$ and the attacker chooses the optimal MPGR, then $\tilde{\mu}_n={\mu}^{Nor}_n+\hat{\mu}^{Mal}_n(m)$, and the estimated normal flow is $\tilde{\mu}_n(m-1)/m$. 
If $x_n=X^{Nor}$, then the estimated normal flow is  $\tilde{\mu}_n$. 
Therefore, the expected normal flow under observation $ \tilde{\mu}_n>0$ is $\zeta(y_n=0,\tilde{\mu}_n>0):=\Pr(x_n=X^{Nor}|y_n=0,\tilde{\mu}_n>0;m)\tilde{\mu}_n+\Pr(x_n=X^{Mal}|y_n=0,\tilde{\mu}_n>0;m)\tilde{\mu}_n(m-1)/m= (\frac{b^{Nor}}{b^{Nor}+b^{Mal}(1-\hat{P}^{Det}(m))} +  \frac{b^{Mal}(1-\hat{P}^{Det}(m))}{b^{Nor}+b^{Mal}(1-\hat{P}^{Det}(m))}\frac{m-1}{m})\tilde{\mu}_n$.} 

We summarize the above flow rate estimation process in Algorithm \ref{algorithm:flow_estimation}, based on the observed {alarm} triggering index $y_0\in \{0,1\}$, individual flow detection index $y_n\in \{0,1\}, n\in \mathcal{N}$, and each device's flow rate $\tilde{\mu}_n\in \mathbb{R}^{0+}, n\in \mathcal{N}$. 

\setlength{\algomargin}{1.07em}
\begin{algorithm}[h]
\SetAlgoLined
\SetKwFunction{FRecurs}{FlowEstimation}
\SetKwProg{Fn}{Function}{:}{}
\Fn{\FRecurs{$\{\tilde{\mu}_n,y_n\}_{n\in \mathcal{N}},y_0$}}{
\eIf{$y_0=1$}{$\tilde{\mu}^{Nor}_n=0, \forall n\in \mathcal{N}$;}
{
\For{$n\leftarrow 1$ \KwTo $N$}
{  \leIf{$y_n=1$ or $\tilde{\mu}_n=0$}{$\tilde{\mu}^{Nor}_n=0$;}{ $\tilde{\mu}^{Nor}_n=\zeta(y_n=0,\tilde{\mu}_n>0)\tilde{\mu}_n$}  }
}}
{
\KwRet 
$\tilde{\mu}^{Nor}=\sum_{n=1}^N \tilde{\mu}^{Nor}_n$ \; 
}
 \caption{Aggregated normal flow rate estimation\label{algorithm:flow_estimation}}
\end{algorithm}


While trying different SFs to learn their impacts, the IoT-NO needs to strike a balance between exploration and exploitation. 
Therefore, we propose an efficient SF learning algorithm in Algorithm \ref{algorithm:UCBbandit} based on the Upper Confidence Bound (UCB) \cite{auer2002finite} and the estimation of the legitimate flow rate in Algorithm \ref{algorithm:flow_estimation}. 
To normalize the reward {(i.e., $\hat{r}_{k}$ in Algorithm \ref{algorithm:UCBbandit})} into the support of $[0,1]$, we add the following two parameters, i.e., $\beta_0=-\min_{m\in \mathcal{M}} c^{SF}(m)$ and a positive constant $\beta_1\in \mathbb{R}^{0+}$ that is sufficiently large. 
{The choice of $\beta_1$ involves a tradeoff between robustness and efficiency. A larger value of $\beta_1$ can accommodate a wider range of $\tilde{\mu}^{Nor}$, thereby better fitting the uncertainty and increasing robustness. 
However, for a given $\tilde{\mu}^{Nor}$, increasing $\beta_1$ will make the reward $\hat{r}_{k}$ less distinct across different SF $m\in \mathcal{M}$, thus reducing learning efficiency. 
Hence, the algorithm is better suited for scenarios where $\tilde{\mu}^{Nor}$ has less uncertainty, allowing for the selection of a relatively small $\beta_1$ based on the estimated value of $\tilde{\mu}^{Nor}$.}

\setlength{\algomargin}{1.07em}
\begin{algorithm}[h]
\SetAlgoLined
 \SetKwFor{ForPar}{for}{do in parallel}{endfor}
\textbf{Initialize} Total learning stage $K$\;
\For{$k=1,2,\cdots,K$}
{
    \leIf{$k\leq M$}{Try each $m\in \mathcal{M}$ once;} 
    {
    $\tilde{m}_t \leftarrow \arg\max_{m\in\mathcal{M}} \tilde{r}_m + \sqrt{2\ln k/k_m}$
    }
    \textbf{Apply} $\tilde{m}_t$ and observe $\{\tilde{\mu}_n,y_n\}_{n\in \mathcal{N}},y_0$\;
    \textbf{Call} Algorithm \ref{algorithm:flow_estimation} to estimate the aggregated normal flow rate: $\tilde{\mu}^{Nor}=\texttt{FlowEstimation}(\{\tilde{\mu}_n,y_n\}_{n\in \mathcal{N}},y_0)$\; 
    \textbf{Compute} the normalized stage reward  $\hat{r}_{k}=[(c^{SF}(\tilde{m}_t)+\beta_0)+(r^{NO}-r^{SP})\tilde{\mu}^{Nor}]/\beta_1$\; 
    \textbf{Update} the average reward $\tilde{r}_{\tilde{m}_t} \leftarrow \frac{k\tilde{r}_{\tilde{m}_t} + \hat{r}_{k}}{k+1}$\;  
    \textbf{Update} the number of selection 
    $k_{\tilde{m}_t} \leftarrow k_{\tilde{m}_t}+1$\; 
}
 \caption{UCB-based efficient SF learning 
 \label{algorithm:UCBbandit}}
\end{algorithm}

\subsection{Price Learning Algorithm for the SAT-SP}
\label{sec:Price Learning Algorithm for the SAT-SP}

Define $\hat{m}(r^{SP})$ as the IoT-NO's best-response SF under $r^{SP}$, i.e., $\hat{m}(r^{SP})\in \arg\max_{m\in \mathcal{M}} \tilde{U}^{SP}(m,r^{SP})$. 
Since the IoT-NO can estimate {$\hat{P}^{Cap}(m,a_0)$}  and $\hat{w}^{Nor}_{1:N}(m)$ through {the} learning process illustrated in Algorithm \ref{algorithm:UCBbandit}, we develop algorithms for the SAT-SP to learn the optimal $r^{SP}$ given the IoT-NO's  best-response SF $\hat{m}(r^{SP})$. 

Let $\hat{U}^{SP}(r^{SP})=\tilde{U}^{SP}(\hat{m}(r^{SP}),r^{SP})$ and  $\hat{U}^{NO}(r^{SP})=\tilde{U}^{NO}(\hat{m}(r^{SP}),r^{SP})$. 
Then we can equivalently reformulate optimization problem \eqref{eq:priceDesign} as 
\begin{equation}
    \label{eq:priceDesign_form2}
         \max_{r^{SP}\in \mathbb{R}^+, \hat{U}^{NO}(r^{SP})\geq 0, \hat{U}^{SP}(r^{SP})\geq 0 } \   \hat{U}^{SP}(r^{SP}). 
\end{equation}

Define shorthand notation $r_{min}=-\frac{r^{Cap}{\alpha^{AT}}}{4({\alpha^{AT}}+{\alpha^{SyR}})}$ and $r_{max}=r^{NO}$. 
According to Remark \ref{remark:IR}, we can restrict the search space of $r^{SP}$ to $[r_{min}, r_{max}]$ for problem \eqref{eq:priceDesign_form2}. 
Compared to the classical feedback learning algorithms to solve  \eqref{eq:priceDesign_form2}, including simulated annealing, Bayesian optimization, and reinforcement learning, that inefficiently search over $[r_{min}, r_{max}]$, we provide an efficient algorithm by exploiting the following \textcolor{black}{piecewise-linear} feature in Lemma \ref{lemma:PLW}.

\begin{lemma} [\textbf{Piecewise-Linear Impact}]
\label{lemma:PLW}
Both $\hat{U}^{NO}$ and $\hat{U}^{SP}$ are \textcolor{black}{piecewise-linear} functions concerning $r^{SP}\in \mathbb{R}^+$ with the same $\bar{M}\leq M$ 
 subdomains. From the left to the right of the $\mathbb{R}^+$, we define these subdomains as $\mathbb{R}_1,\dots \mathbb{R}_{\bar{M}}$, respectively. 
 The values of $\tilde{U}^{NO}(m,0)$ and $\tilde{U}^{SP}(m,0)$ increase and decrease concerning $m\in \mathcal{M}$, respectively.  
\begin{itemize}
    \item [(a).] For the IoT-NO, $\hat{U}^{NO}$ is further a convex and decreasing function concerning $r^{SP}\in \mathbb{R}^+$. 
    From subdomain $\mathbb{R}_i$ to subdomain $\mathbb{R}_j$, if $i<j$ and $i,j\in \{1,\cdots,\bar{M}\}$, then the best response SF $\hat{m}(r^{SP})\in \mathcal{M}$ decreases.  
    \item [(b).] 
    For the SAT-SP, $\hat{U}^{SP}$ is possibly discontinuous at the boundaries of different subdomains.  
    In each subdomain $\mathbb{R}_i, i\in \{1,\cdots,\bar{M}\}$, the value of $\hat{U}^{SP}(r^{SP})$ increases with $r^{SP}$. 
    Moreover, in subdomain $\mathbb{R}_j, j>i, j\in \{1,\cdots,\bar{M}\}$, the value of $\hat{U}^{SP}(r^{SP})$ increases with $r^{SP}$ at smaller rates than it does in subdomain $\mathbb{R}_i$.  
\end{itemize}
\end{lemma}
\begin{proof}
    We first prove statement (a). Following the definition, $\hat{U}^{NO}(r^{SP})=\max_{m\in \mathcal{M}} \tilde{U}(m,r^{SP})$. 
    Since $\tilde{U}^{NO}(m,r^{SP})$ is an affine and decreasing function over $r^{SP}$ for any $m\in \mathcal{M}$, the point-wise maximum of a group of affine functions over $r^{SP}$ leads to a piecewise-linear, convex, and decreasing function with at most $M$ subdomains. 
    Since the slope of $\tilde{U}^{NO}(m,r^{SP})$ concerning $r^{SP}$ 
    decreases as $m$ increases, the best response SF $\hat{m}(r^{SP})$ decreases as we shift to subdomains with a larger index. 
    We then prove statement (b). The piecewise linearity and increase of $\hat{U}^{SP}$ over $r^{SP}$ follows from the fact that $\tilde{U}^{SP}(m,r^{SP})$ is an affine and increasing function concerning $r^{SP}$. 
    Since the subdomain is determined by the NO's best response $\hat{m}(r^{SP})$, $\hat{U}^{SP}$ is possibly discontinuous at the boundaries of these subdomains.    
    Since the best response SF $\hat{m}(r^{SP})$ decreases as we shift to subdomains with a larger index, and the slope of $\tilde{U}^{SP}(m,r^{SP})$ concerning $r^{SP}$ 
    decreases as $m$ decreases, the value of $\hat{U}^{SP}(r^{SP})$ increases with $r^{SP}$ at smaller rates in subdomains of larger indexes. 
\end{proof}

We illustrate the \textcolor{black}{piecewise-linear} property in Fig. \ref{fig:PWL}. 
\begin{figure}[h]
\centering
\includegraphics[width=.7 \columnwidth]{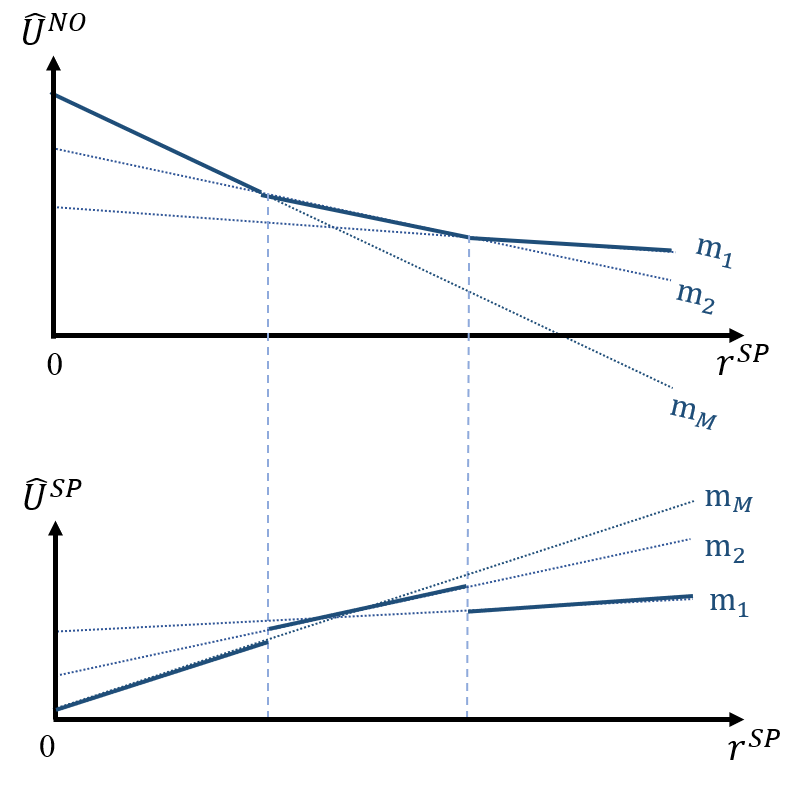}
\caption{ 
The piecewise linear impact of $r^{SP}$ on $\hat{U}^{NO}$ and $\hat{U}^{SP}$ with $\bar{M}=3$ subdomains separated by two horizontal lines. 
}
\label{fig:PWL}
\end{figure}
The dashed lines illustrate $\tilde{U}^{NO}(m,r^{SP})$ and $\tilde{U}^{SP}(m,r^{SP})$, and the solid lines illustrate $\hat{U}^{NO}(r^{SP})$ and $\hat{U}^{SP}(r^{SP})$. 
The two horizontal lines separate the $\bar{M}=3$ subdomains. 
Based on Lemma \ref{lemma:PLW}, Remark \ref{remark:Impact of the Flow Transmission Fee} summarizes the impact of the flow transmission fee.  

\begin{remark}[\textbf{Impact of the Flow Transmission Fee}]
\label{remark:Impact of the Flow Transmission Fee}
 Under a given configuration, if the SAT-SP increases (resp. decreases) the flow transmission fee $r^{SP}$, then the IoT-NO tends to reduce (resp. increase) {the} SF.  
\end{remark}


The SAT-SP does not know the explicit form of the IOT-NO's best-response SF $\hat{m}(r^{SP})$, but can estimate the value of $\hat{U}^{SP}(r^{SP})$ for each $r^{SP}$. 
Therefore, {the SAT-SP's} learning goal is to efficiently search the entire feasible region $[r_{min}, r_{max}]$ to find the optimal $r^{SP}$. 
Based on Lemma \ref{lemma:PLW}, we can turn the above region search problem into determining the change points of the subdomains, as illustrated in Algorithm \ref{algorithm:rSPlearning}. 
Established based on bisection search, Algorithm \ref{algorithm:Determineslope} learns the slope of the leftmost subdomain, and  Algorithm \ref{algorithm:changeDetection} learns the change point set for $\bar{M}$ subdomains, where 
$S(x_1,x_2):=\frac{\hat{U}^{SP}(x_2)-\hat{U}^{SP}(x_1)}{x_2-x_1}$ represents the slope of the line determined by two values of $r^{SP}$ (i.e., $x_1$ and $x_2$). 

\setlength{\algomargin}{1.07em}
\begin{algorithm}[h]
\SetAlgoLined
\small 
\SetKwFunction{FRecurs}{SL}
\SetKwProg{Fn}{Function}{:}{}
\Fn{\FRecurs{$r^L,r^R,T_S$}}{\While{True}
{
 $r^M\leftarrow (r^L+r^R)/2$; 
 Add $r^M$ to the set $SAMPLES$\;
 \eIf{
 $|S(r^L,r^R)-S(r^L,r^M)|<T_S$
 }{
 $SLOPE\leftarrow S(r^L,r^M), RSP\leftarrow r^R$\;
 \textbf{Break}\;
 }
 {$r^R\leftarrow r^M$}
}}{\KwRet $SLOPE, RSP, SAMPLES$\;}
 \caption{Leftmost Subdomain Slope Learning
 \label{algorithm:Determineslope}}
\end{algorithm}

\setlength{\algomargin}{1.07em}
\begin{algorithm}[h]
\SetAlgoLined
\small 
\SetKwFunction{FRecurs}{CD}
\SetKwProg{Fn}{Function}{:}{}
\Fn{\FRecurs{$r^L,r^R,SL,T_C,T_S$}}{
$CP_L\leftarrow r^L$\; 
\While{True}
{
 $r^M \leftarrow (r^L+r^R)/2$; 
 Add $r^M$ to the set $SAMPLES$\;
 \lIf{$r^M-r^L<T_C$}{\textbf{Break}}
 \eIf{
 $|S(r^L,r^M)-SL|<T_S$
 }{
 $r_m^L\leftarrow r_m^M, CP_L \leftarrow r_m^M$\;
 }
 {$r_m^R\leftarrow r_m^M,CP_R \leftarrow r_m^M$\;}
}}{\KwRet $CP_L, CP_R, SAMPLES$\;}
 \caption{Subdomain Change Points Learning
 \label{algorithm:changeDetection}}
\end{algorithm}

In Algorithm \ref{algorithm:rSPlearning}, $RSP_m$, $SL_m$, and $BP_m$ represent the rightmost sample point, the slope, and the break point for each subdomain $m$, respectively. 
Denote the total number of break points found as $BP^{Num}$. 
Denote $SPF_m$ as the number of sampled points in the set SAMPLES that belongs to the $m$-th subdomain. 
The tolerance levels for the change point and the slope are denoted as $T_C$ and $T_S$, respectively. 
Lines 2-14 in algorithm \ref{algorithm:rSPlearning} construct the estimated piecewise-linear function $V^{SP}$ for the general cases, while lines 15-19 further provide a verification to eliminate the extreme cases of collinearity, i.e., multiple points on one line. 
In lines 5-6, if $RSP_m=r_m^R$, then $\hat{U}^{SP}$ is linear over $r^{SP}$ for the entire search space $[r_{min},r_{max}]$. We can terminate the subsequent computation to enhance efficiency.

\setlength{\algomargin}{1.07em}
\begin{algorithm}[h]
\SetAlgoLined
\small 
 \SetKwFor{ForPar}{for}{do in parallel}{endfor}
\textbf{Initialize} $r_{1}^L=r_{min},r_{1}^R=r_{max}$, $SAMPLES=\{r_{1}^L,r_{1}^R\}$\; 
\For{$m=1,2,\cdots,M$}
{
Call Algorithm \ref{algorithm:Determineslope} to learn the slope of the $m$-th subregion: 
$SL_m,RSP_m,SAM_{SL}=\texttt{SL}(r_m^L,r_m^R,T_S)$\; 
$SAMPLES\leftarrow SAMPLES \cup SAM_{SL}$\;
\eIf{$RSP_m=r_m^R$}{\textbf{Break}}
{$r_m^L\leftarrow RSP_m$\;  
Call Algorithm \ref{algorithm:changeDetection} to learn the $m$-th change point: 
$CP_L,CP_R,SAM_{CD}=\texttt{CD}(r_m^L,r_m^R,SL_m,T_C,T_S)$; 
$BP_m=(CP_L+CP_R)/2, BP^{Num}\leftarrow BP^{Num}+1$\;
$SAMPLES\leftarrow SAMPLES \cup SAM_{CD}$\;
}
$r_m^L\leftarrow CP_R$\;
}
Construct the estimated piecewise-linear function $V^{SP}$ for $r^{SP}\in [r_{min},r_{max}]$ based on $\{BP_m,SL_m\}_{m=1,\cdots,BP^{Num}}$\; 
\If{$BP^{Num}<M$}
{
\For{$m=1:BP^{Num}$}{
\lIf{$SPF_m<M$}
{Sample $M-SPF_m$ points in subdomain $m$ and verify their belongings to $V^{SP}$
}
}
}
 \caption{Adaptive Price Sampling Method 
 \label{algorithm:rSPlearning}}
\end{algorithm}


\section{Performance Evaluation}
\label{sec:case study}

{We design a set of simulations to illustrate the security loss inflicted by the cross-domain attack, the optimality of the proposed multi-domain collaborative defense, and the effectiveness of the learning algorithms.} 

\subsection{{Device-Level Compromise and Individual Flow Sampling}}
{As illustrated in Fig. \ref{fig:SimulationSetting}, we consider $N=100$ sensors and $W=1000$ communication sessions, each potentially lasting from hours to days. 
During each session, sensor $n\in \mathcal{N}$ could be in one of the three states (idle $X^{Idl}$, normal $X^{Nor}$, or compromised $X^{Mal}$), as governed by the transition matrix in \eqref{eq:Tr}, with parameters  $\alpha^{SeR}=0.8,\alpha^{AT}=0.02,\alpha^{SyR}=0.15,\alpha^{SeT}=0.3$. 
Each non-idle sensor $n\in \mathcal{N}$ generates legitimate packets at an average rate of $\mu_n^{Nor}=100$ packets per session (pps). 
Additionally, any compromised sensor generates additional malicious packets at an average rate of $\mu_n^{Mal}=15$ pps.}   
 

\begin{figure}[h]
\centering
\includegraphics[width=1 \columnwidth]{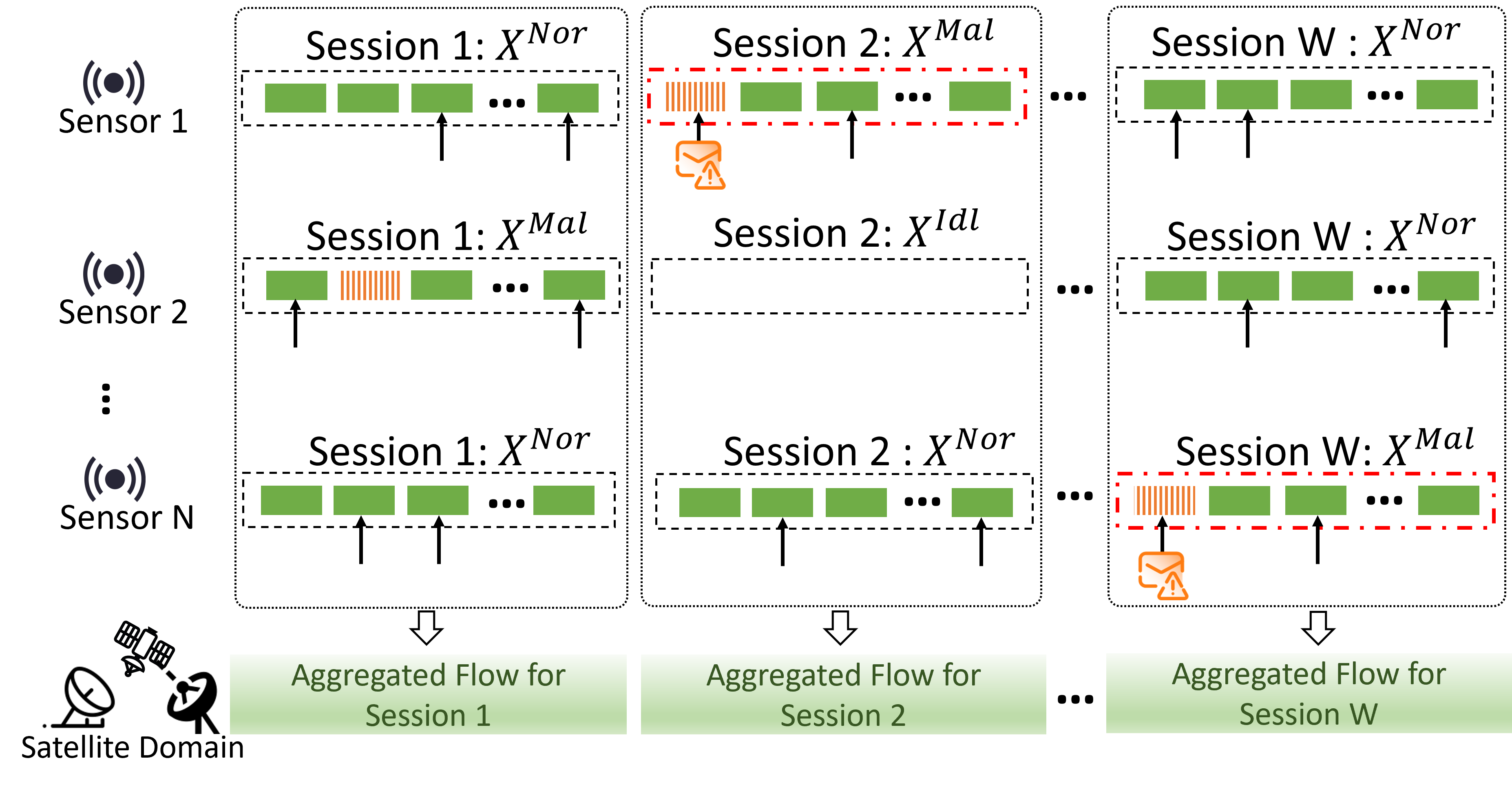}
\caption{ 
{
An overview of the simulation setting. 
During each session, the green squares represent legitimate packets, while orange squares with vertical stripes denote malicious packets. 
The arrows indicate the samples selected by the IoT-NO. 
The aggregated flow from each session, depicted at the bottom, enters the satellite domain.}  
}
\label{fig:SimulationSetting}
\end{figure}

{During each session, the IoT-NO samples $m\in \mathcal{M}$ times and terminates the session once a malicious packet has been sampled and correctly identified by the DPI, as indicated by the two exclamation signs for session $2$ of sensor $1$ and session $W$ of sensor $N$ in Fig. \ref{fig:SimulationSetting}. 
We have extended the setting to include a miss detection rate for the DPI, denoted by $\psi\in [0,1]$. 
We create the color-coded device-session detection matrices in Fig. \ref{fig:matrix3by1_u15} to illustrate the sampling process at the IoT-NO with  $\psi=1$ under $m=2$ and $m=4$. 
Each matrix consists of $N$ by $W$ squares. 
The color of the square at column $n\in\{1,\cdots,N\}$ and row $w\in\{1,\cdots,W\}$, yellow or blue, indicates whether session $w$ of device $n$ has been terminated or not. 
As demonstrated by the increasing number of yellow squares in the right-hand color matrix, a higher SF enhances the detection and deterrence of malicious packets at the individual-flow level. 
}  

\begin{figure}[h]
\centering
\includegraphics[width=1 \columnwidth]{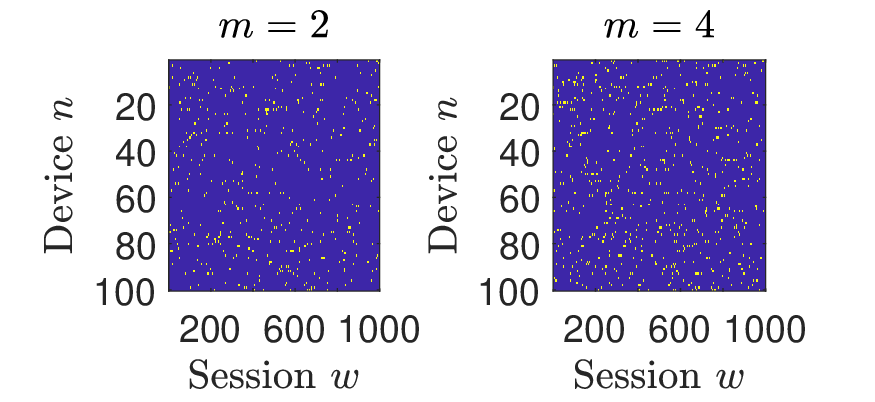}
\caption{
{Device-session detection matrices across $N$ devices and $W$ sessions for the sampling process at the IoT-NO with $\psi=1$ under $m=2$ and $m=4$. 
A yellow square at column $n\in\{1,\cdots,N\}$ and row $w\in\{1,\cdots,W\}$ indicates the termination of session $w$ for device $n$.}  
}
\label{fig:matrix3by1_u15}
\end{figure}

\subsection{{Honeypot Triggering based on the Aggregated Flow}}
{Except for the packets in sessions terminated by the IoT-NO, all remaining packets comprise the aggregated flow entering the satellite domain, as depicted at the bottom of Fig. \ref{fig:SimulationSetting}. 
The SAT-SP can choose between $K=2$ configurations $\mathcal{A}=\{a^L,a^H\}$, where $a^L$ and $a^H$ denote low-interaction and high-interaction honeypots, respectively. 
Compared to low-interaction honeypots that emulate critical components at a low cost, high-interaction honeypots utilize real resources to more effectively attract attacks, albeit at a higher cost \cite{7676152}. 
Consequently, we set the probability that each malicious packet triggers the honeypot to $p^L=0.01$ for configuration $a^L$, and $p^H=0.03$ for configuration $a^H$, respectively.} 

{ 
Once a malicious packet in the aggregated flow of session $w \in \mathcal{W}$ triggers a honeypot in the satellite domain, the satellite terminal terminates the aggregated flow of that session. 
Under the realization of the individual-level flow sampling process depicted in Fig. \ref{fig:matrix3by1_u15}, we visualize the honeypot triggering probability and the termination of the aggregated flow across $W$ sessions in Fig. \ref{fig:matrixplot} and Fig. \ref{fig:HoneyTrigMatrixShow}, respectively. 
The results confirm that individual flow sampling with a higher SF effectively reduces the probabilities of triggering honeypots and terminating aggregated flows by actively preventing malicious packets from entering the satellite domain.
Meanwhile, a high-interaction configuration is more effective in capturing malicious packets. 
} 

\begin{figure}[htb]
    \centering 
\begin{subfigure}{0.49\columnwidth} 
  \includegraphics[width=1 \linewidth]{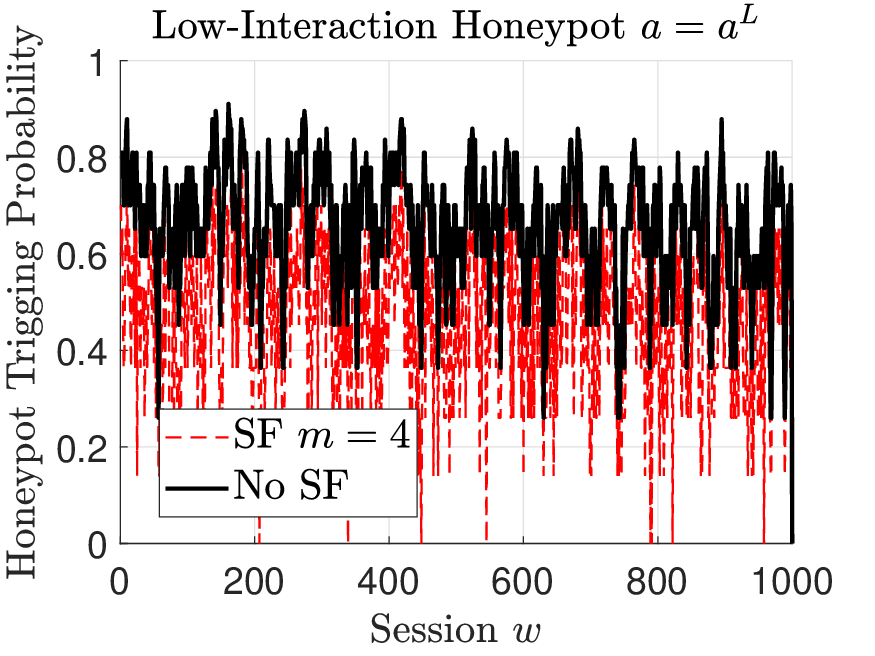}
  \caption{Low-Interaction $a=a^L$.} 
  \label{fig:matrixplot_m2}
\end{subfigure}\hfil
\begin{subfigure}{0.49\columnwidth} 
  \includegraphics[width=1 \linewidth]{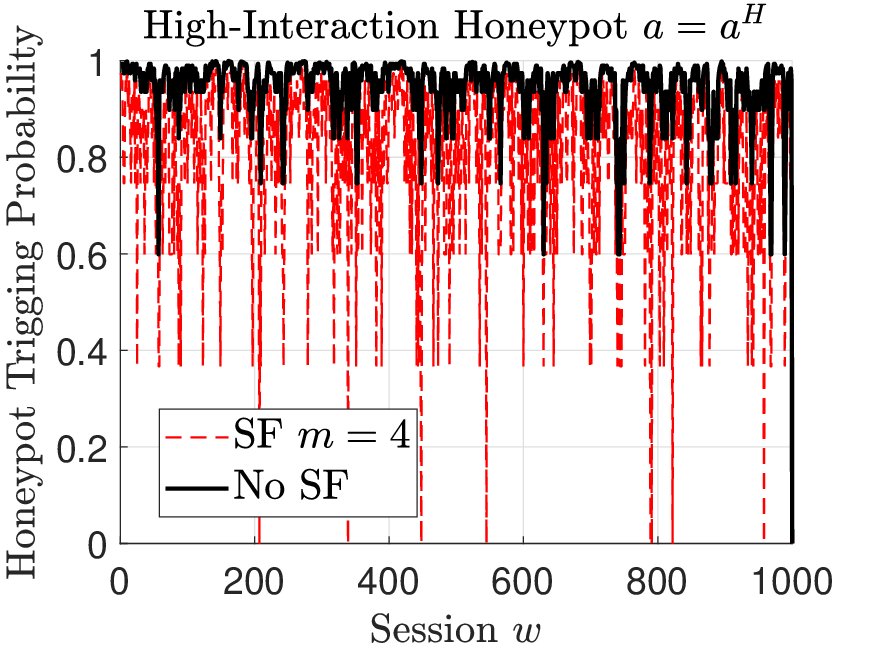}
  \caption{High-Interaction $a=a^H$. } 
  \label{fig:matrixplot_m4} 
\end{subfigure}\hfil 
\caption{
{Honeypot triggering probabilities with and without individual flow sampling in dashed red and solid black lines, respectively.}}  
\label{fig:matrixplot}
\end{figure}

\begin{figure}[h]
\centering
\includegraphics[width=1 \columnwidth]{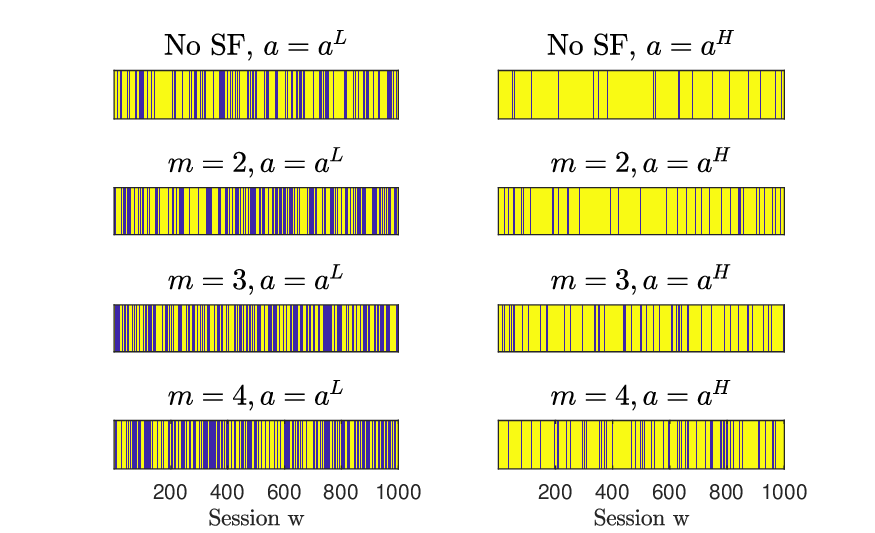}
\caption{{Aggregate flow termination under various SFs and honeypot configurations across $W$ sessions, with yellow indicating termination and blue indicating non-termination.}}
\label{fig:HoneyTrigMatrixShow}
\end{figure}



\subsection{{Security Loss under Attacks and Defense Mitigation}}
\label{sec:security loss}
{In this section, we illustrate the security loss caused by cross-domain attacks and the mitigation effects of cooperative defense.}   
The SAT-SP adopts configuration $a^L$ with $r^{Cap}=0$, $r^{SP}=0.2$, $r^{NO}=0.3$. 
We set the costs $c^{SF}(m)=0$ for all $m \in \mathcal{M}$ and $c^{OC}(a)=0$, and postpone the discussion of incentives for the attacker, the IoT-NO, and the SAT-SP to Section \ref{sec:Gambling Effect and Security Gain under CSNE}.

{
We define the attack intensity as $\eta\in [0,1]$ and set MPGR $\mu^{Mal}_n=\eta\mu^{Nor}_n$ for all $n\in \mathcal{N}$ to investigate its impact in Fig. \ref{fig:attackIntensity}. 
The impacts of MPGR on the SAT-SP's and the IoT-NO's utilities are represented by red and blue lines, respectively.} 
As illustrated by the dashed lines, when the IoT-NO applies no flow sampling, then the utilities of both the IoT-NO and the SAT-SP decrease rapidly as attack intensity $\eta\in [0,1]$ increases. 
\textit{Thus, the honeypots at the satellite domain itself are inadequate for defending against cross-domain attacks{, particularly when the attack intensity is high.}}  
When the IoT-NO applies SF $m=2,3,4$, respectively,  {the decreases in utility caused by the attacks} are largely mitigated at low values of $\eta$. 
{Moreover, a larger value of $\eta$ also increases the individual-flow detection rate, reduces the honeypot triggering probability, and promotes communication transmission, especially for high SFs.} Consequently, the utilities begin to increase {(though still below those in an attack-free scenario, where $\eta=0$)} when the attack intensity $\eta$ exceeds a certain threshold. 
\textit{These results illustrate the benefits of cooperative defense to counter cross-domain attacks.} 


\begin{figure}[h]
\centering
\includegraphics[width=1 \columnwidth]{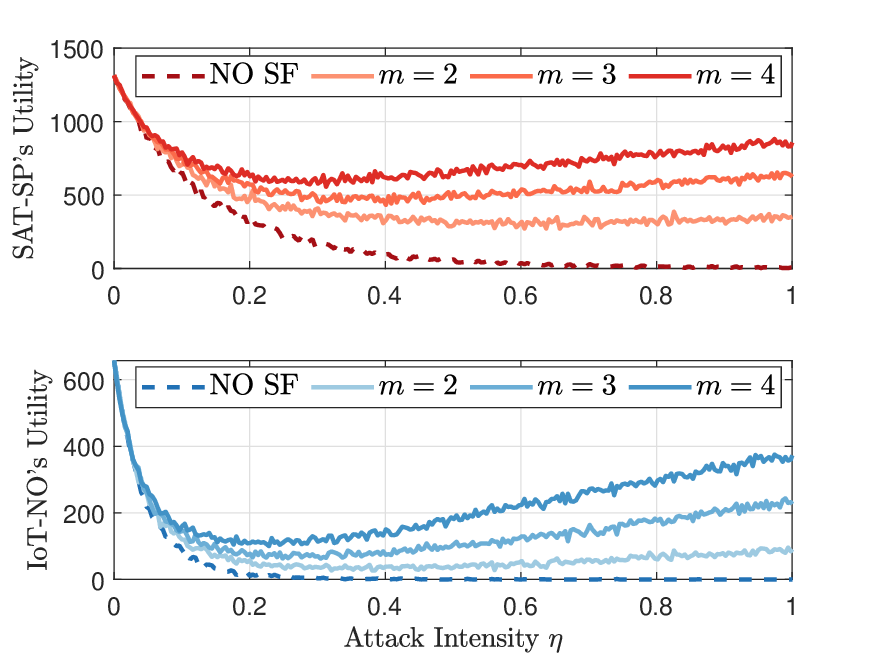}
\caption{ 
{Impact of attack intensity $\eta\in [0,1]$ on} the utilities of the IoT-NO and the SAT-SP with or without individual flow sampling.
} 
\label{fig:attackIntensity}
\end{figure}

We investigate the impact of the device's scalability on the IoT-NO and the SAT-SP in In Fig. \ref{fig:N_IOT_NO} and \ref{fig:N_SAT_SP}, respectively. 
The black dashed lines in both figures illustrate the utilities {in the absence of attacks, i.e., $\eta=0$. Serving as benchmark values, these lines} demonstrate that both the IoT-NO and the SAT-SP can benefit from serving more devices, referred to as the \textit{Profit of Scale (PoS)}. 
Next, we consider the case where attacks exist with intensity $\eta=0.15$. 
When the IoT-NO applies no flow sampling, as illustrated by the dotted lines in Fig. \ref{fig:Nnumber}, then both utilities of the IoT-NO and the SAT-SP start to decrease when the number of devices $N$ has passed a certain threshold. 
Compared with the benchmark values in black, we observe that \textit{the cross-domain attack has invalidated the PoS}. 
However, \textit{the individual flow sampling at the IoT gateway with SF $m=2,3,4$ can {significantly  mitigate the PoS degradation caused by the attacks}}. 
\begin{figure}[htb]
    \centering 
\begin{subfigure}{0.49\columnwidth}
  \includegraphics[width=\linewidth]{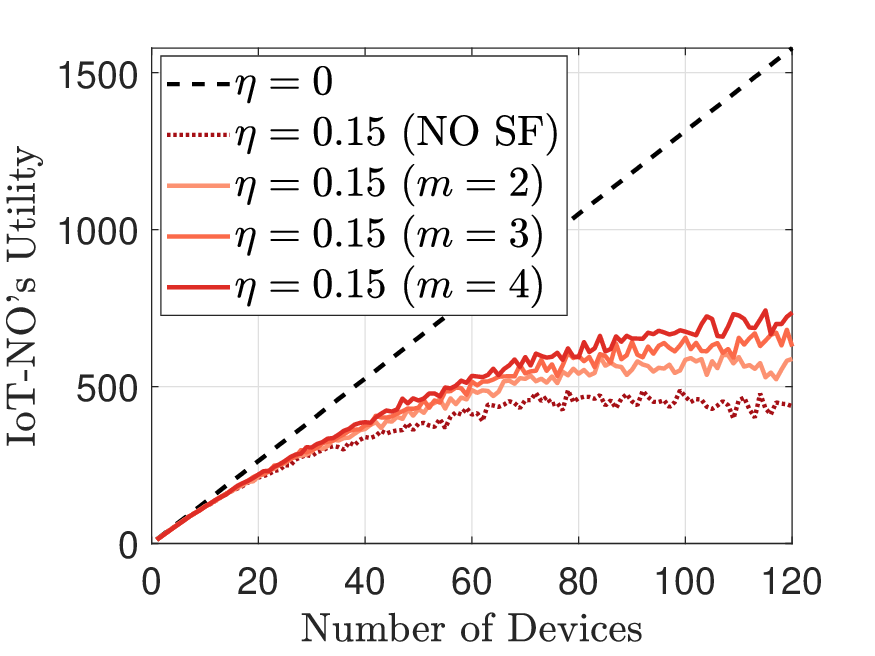}
  \caption{IoT-NO's expected utility.} \label{fig:N_IOT_NO}
\end{subfigure}\hfil 
\begin{subfigure}{0.49\columnwidth}
  \includegraphics[width=\linewidth]{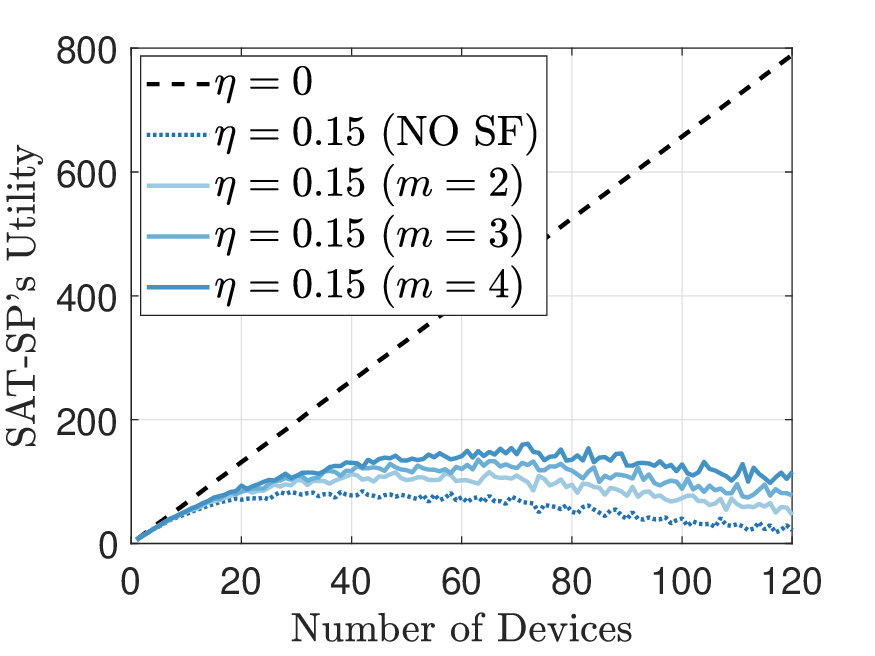} 
  \caption{SAT-SP's expected utility. } \label{fig:N_SAT_SP} 
\end{subfigure}\hfil 
\caption{Utilities of the IoT-NO and the SAT-SP for different numbers of devices with or without flow sampling at the IoT gateway. }  
\label{fig:Nnumber}
\end{figure}

\subsection{Security Gain under CSNE}
\label{sec:Gambling Effect and Security Gain under CSNE}
We consider $M=3$ SFs ($m_1=2,m_2=3$, and $m_3=4$) and investigate the mixed-strategy CSNE under the conditions $c^{SF}(m_1)=-1$ and $c^{SF}(m_2)=-10$, {with $c^{SF}(m_3)$ varying between $[-60,-40]$}. 
As illustrated in Fig. \ref{fig:EqStrategy}, the IoT-NO's CSNE strategy falls into three regions as $c^{SF}(m_3)$ increases. 
In region one, where $c^{SF}(m_3)\leq -59.6$, SF $m_2$ dominates other SFs due to its cost-effectiveness. 
{In region two, where $-59.6 \leq c^{SF}(m_3)\leq -52.3$, the IoT-NO chooses SF $m_2$ with probability $0.14$ and $m_3$ with probability $0.86$, responding to the decreased cost of taking SF $m_3$. When $c^{SF}(m_3) \geq -52.3$, the IoT-NO further increases the probability of selecting $m_3$ to $0.91$ and substitutes $m_2$ with $m_1$.}  

In regions one and two, we have observed a \textit{gambling effect} in the SAT-SP's strategy, illustrated by the dotted black line. 
Given the common knowledge of the decreased cost of  selecting SF $m_3$, the SAT-SP can reduce security effort (by taking $a^H$ with lower probability) within a certain range. 
{In this manner, the SAT-SP prompts the IoT-NO to continue its security efforts by maintaining sample frequency $m_3$.} 

\begin{figure}[h]
\centering
\includegraphics[width=.9 \columnwidth]{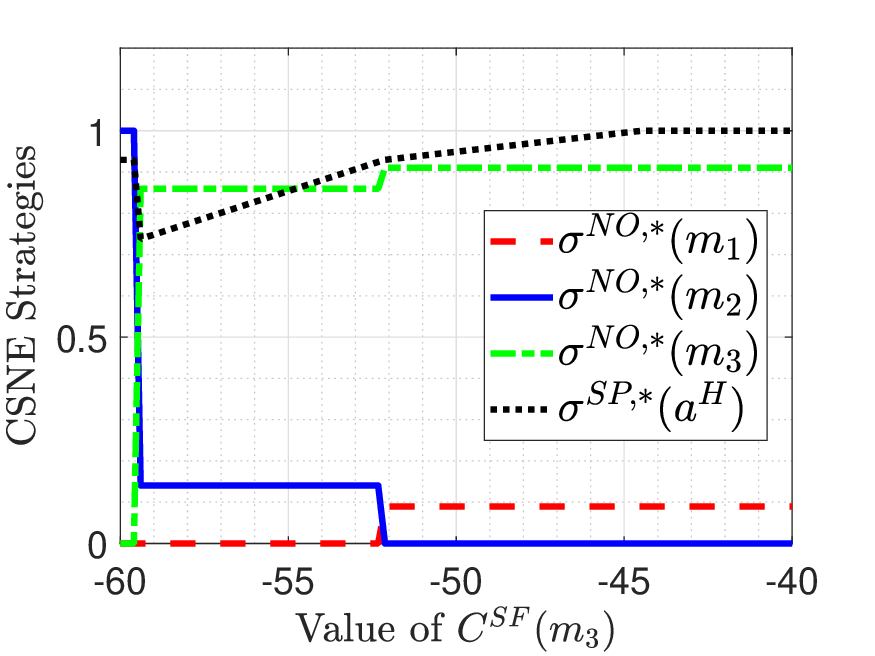}
\caption{ 
The CSNE strategies of the IoT-NO and the SAT-SP.  
}
\label{fig:EqStrategy}
\end{figure}

In Fig. \ref{fig:valueCSNE}, we illustrate the security gains of players following the CSNE, denoted by the dotted black lines, under varying costs of $c^{SF}(m_3)$. 
In Fig. \ref{fig:UtilityDeviate_NO}, the SAT-SP remains to take $\sigma^{SP,*}$, while the IoT-NO chooses different SFs. 
The results show that CSNE can bring {a maximum} security gain of $11.4\%,  11.9\%$, and $9.5\%$ to the IoT-NO who remains the equilibrium $\sigma^{NO,*}$, compared to the IoT-NO who deviates to SF $m_1,m_2$, and $m_3$, respectively. 
In Fig. \ref{fig:UtilityDeviate_SP}, the IoT-NO remains to take $\sigma^{NO,*}$, while the SAT-SP chooses different honeypot configurations. 

\begin{figure}[htb]
    \centering 
\begin{subfigure}{0.49\columnwidth}
  \includegraphics[width=\linewidth]{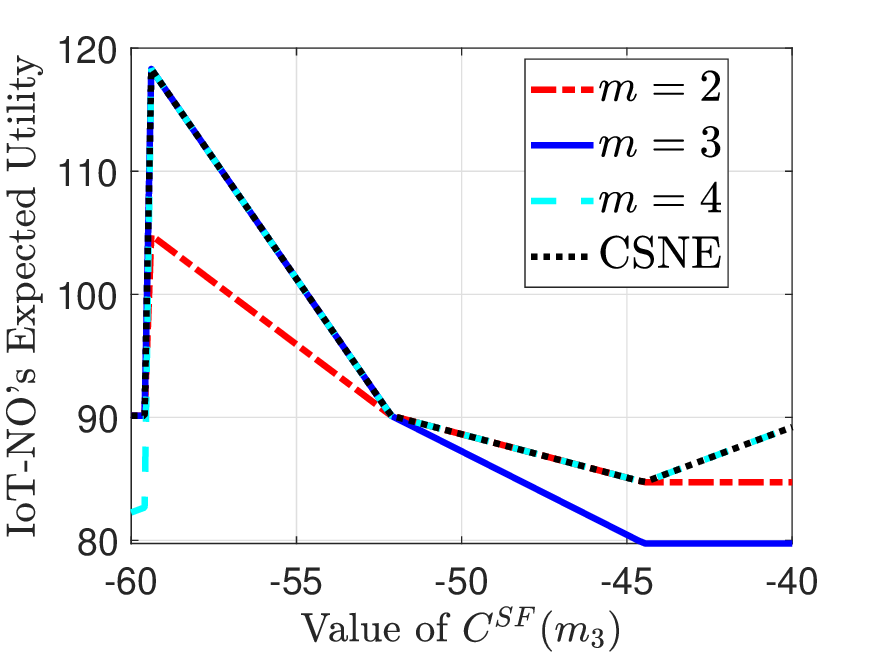}
  \caption{
IoT-NO's expected utility. 
 }
 \label{fig:UtilityDeviate_NO} 
\end{subfigure}\hfil 
\begin{subfigure}{0.49\columnwidth}
  \includegraphics[width=\linewidth]{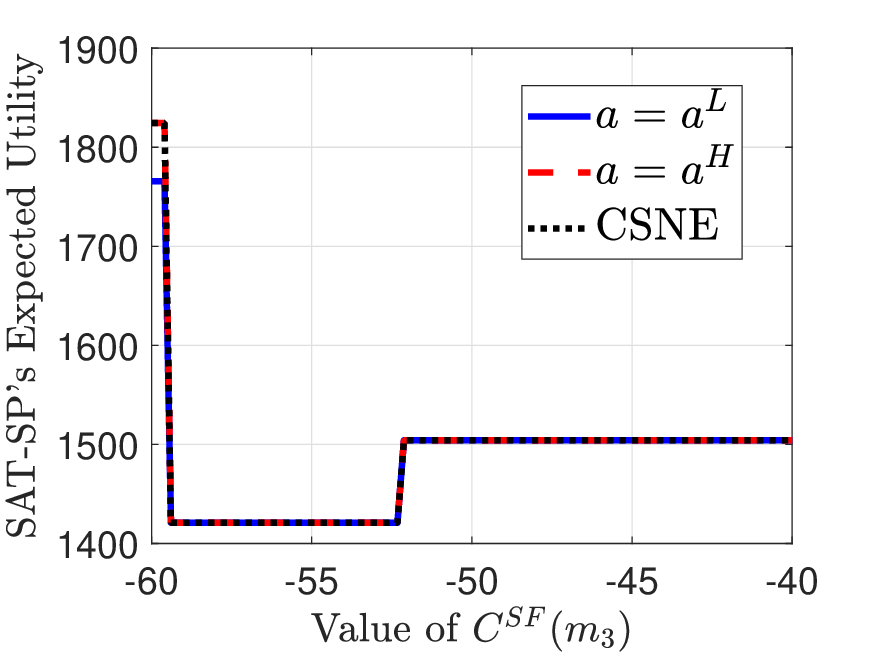} 
  \caption{
SAT-SP's expected utility. 
  }
  \label{fig:UtilityDeviate_SP} 
\end{subfigure}\hfil 
\caption{Expected utilities {under CSNE and deviated strategies.}}  
\label{fig:valueCSNE}
\end{figure}

\subsection{Learning Processes for the IoT-NO and the SAT-SP}
\label{sec:Learning Processes for the IoT-NO and the SAT-SP}
In this section, we illustrate the IoT-NO's learning process for a given $r^{SP}=0.2$ when Algorithm \ref{algorithm:UCBbandit} is applied. 
Define the per-period regret at learning stage $k$ as $\lambda^{Per}_k:=\max_{m\in \mathcal{M}} \tilde{U}^{NO}(m,r^{SP})-\tilde{U}^{NO}(\tilde{m}_k,r^{SP})$, where $\tilde{m}_k$ is the selected SF at stage $k$. 
Then, we can define the average regret at learning stage $k$ as $\lambda^{Avg}_k:=\sum_{k'=1}^{k} \lambda^{Per}_{k'}/k$. 
As illustrated in Fig. \ref{fig:PLOT_regret}, the average regret in all five trials decreases with the learning stages efficiently. 

\begin{figure}[h]
\centering
\includegraphics[width=.8 \columnwidth]{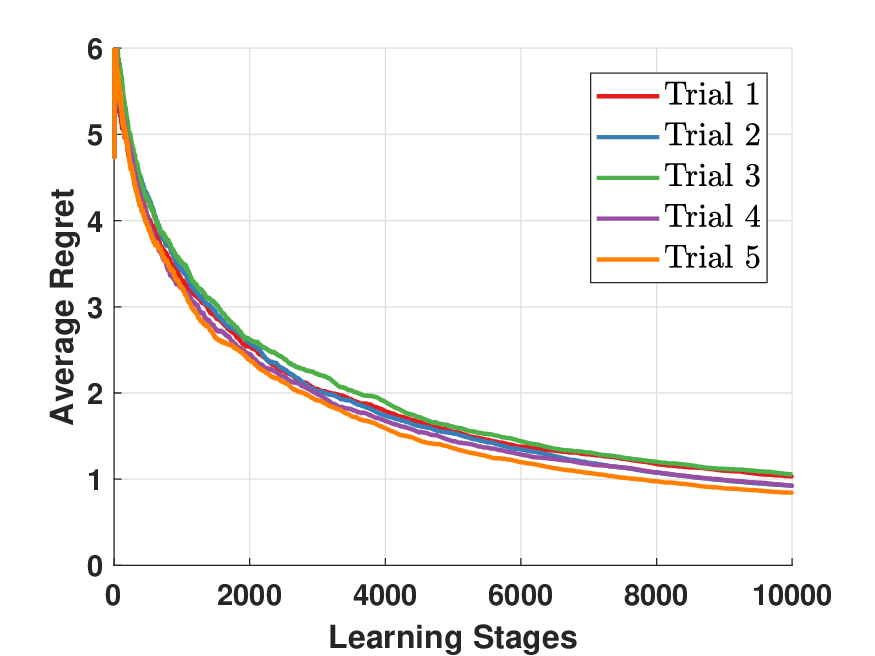}
\caption{
The decrease in the average regret under five different trials. 
}
\label{fig:PLOT_regret}
\end{figure}


We illustrate the SAT-SP's price learning process in Fig. \ref{fig:pricelearning} under different tolerance levels.  
The ground-truth value of the learning objective function $\hat{U}^{SP}$ shown in red corroborates the PLW property in Lemma \ref{lemma:PLW}. 
We illustrate the samples in blue circles, where the size (resp. the gradation of color) of a sample is (resp. inverse proportional) proportional to the sequence of the sample. A video demonstration is available at \url{https://cloud.tsinghua.edu.cn/f/f2552ed4159849e3b06a/}. 

\begin{figure}[htb]
    \centering 
\begin{subfigure}{0.49\columnwidth}
  \includegraphics[width=\linewidth]{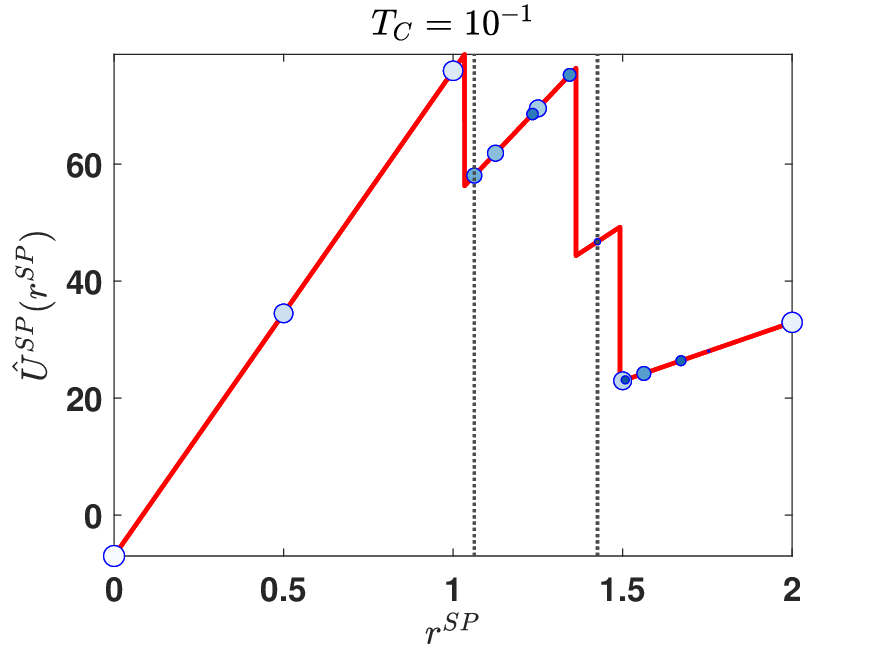}
  \caption{
$T_C=10^{-1}$. 
 }\label{fig:e-1} 
\end{subfigure}\hfil 
\begin{subfigure}{0.49\columnwidth}
  \includegraphics[width=\linewidth]{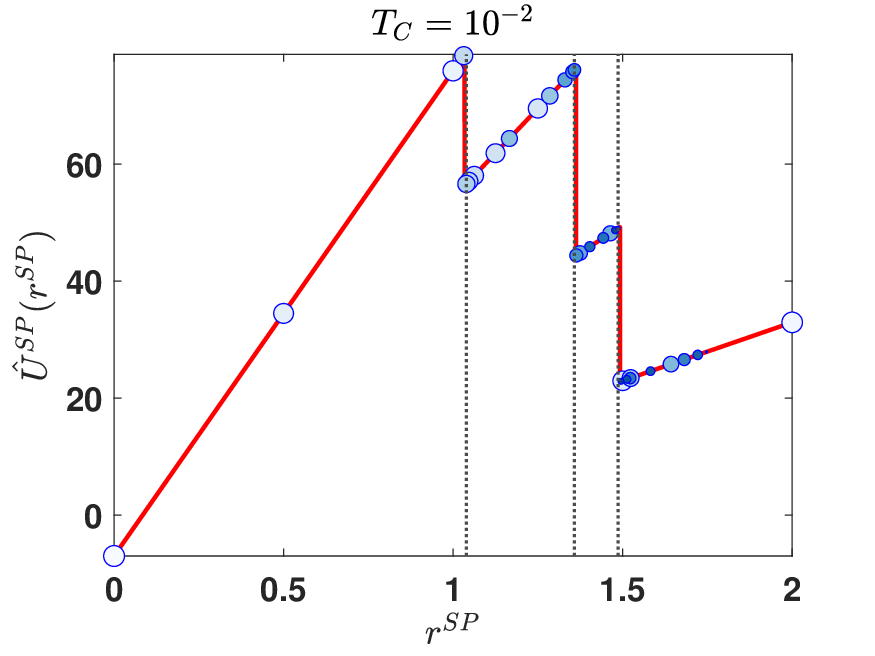} 
  \caption{
$T_C=10^{-2}$. 
  }\label{fig:e-2} 
\end{subfigure}\hfil 
\caption{The SAT-SP's price learning process. 
}  
\label{fig:pricelearning}
\end{figure}

As illustrated in Fig. \ref{fig:e-2}, when we select a tolerance level that is sufficiently small (e.g., $T_C=10^{-2}$), we can learn all the subdomains within $32$ points. 
If the selected tolerance level is large (e.g., $T_C=10^{-2}$), as illustrated in Fig. \ref{fig:e-1}, then the learning process only needs to sample $15$ points. However, it misses the third change point as the third and fourth change points are too close for the tolerance level to distinguish from. 

\section{Conclusion and Discussions}
\label{sec:conclusion}
\subsection{\textcolor{black}{Conclusion}}
This work {contributes to the field of satellite cybersecurity by developing} a consolidated game framework for satellite-enabled IoT to comprehensively model the cross-domain attack and defense at the IoT devices (i.e., the dynamic compromise of IoT devices), the IoT access network (i.e., the IoT-NO's individual flow sampling), and the satellite transmission network (i.e., the SAT-SP's {coarse-grid defense} configurations). 
We have introduced the equilibrium concept of CSNE, proved its existence, and formulated a bilinear program for its computation. 
Our analysis of the cross-domain impact has yielded structural results and theoretical insights, including a threshold division of the SF set and the identification of fundamental limits for the security-efficiency ratio of dominant strategies. 
{These findings can guide the efficient design of the SAT-SP's configuration.} 
Recognizing the price design problem as NP-hard, we have equated the IoT-NO's best-response SF to a bandit arm problem and proposed a UCB-based SF learning algorithm. 
Leveraging the \textcolor{black}{piecewise-linear} property, we have developed an efficient price learning algorithm for the SAT-SP that is guaranteed to {achieve the requested accuracy within a finite number of steps. 
The identified role of this property in enhancing learning efficiency can be applied to other learning scenarios with similar structures.} 

{A series of simulations has demonstrated that single-domain defense within the satellite network is inadequate against cross-domain attacks, while highlighting the benefits of cooperative defense with individual flow sampling at the IoT gateway, particularly under high attack intensity.  
The numerical results concerning the number of IoT devices has demonstrated that cooperative defense can significantly mitigate the PoS degradation caused by cross-domain attacks. 
The maximum utility increase of $11.9\%$ under CSNE, compared to those under deviated strategies, demonstrates the effectiveness of the consolidated game in finding optimal defense strategies among participants with non-aligned incentives.} 

\subsection{\textcolor{black}{Discussions of Privacy and Advanced Attack Scenarios}}
\textcolor{black}{
While primarily focusing on security, this work also considers the privacy implications between the IoT-NO and the SAT-SP, which motivate the development of a multi-domain defense strategy in Satellite IoT systems. 
However, the current framework’s requirement for individual flow sampling at the IoT gateway introduces two additional privacy concerns. First, the IoT-NO can fully observe the data flows of IoT devices, potentially compromising device-level privacy. Second, the SAT-SP needs to share features of malicious packets with the IoT-NO, which could reveal sensitive threat intelligence. 
To mitigate the first privacy concern, privacy-protection mechanisms such as anonymization and indistinguishability can be implemented to protect the identities of IoT devices when interacting with the IoT-NO. For the second privacy concern, differential privacy measures can be adopted to ensure the SAT-SP's threat features are shared in a privacy-preserving manner, balancing privacy with the effectiveness of detecting malicious packets.
These enhancements further strengthen the framework’s consideration of privacy alongside security. 
}

\textcolor{black}{
This work focuses on the injection of malicious packets with the ultimate goal of compromising satellites, providing a foundational understanding of the risks and mitigation strategies for cross-domain attacks. Building on this foundation, future work will investigate more advanced attack scenarios.
First, beyond detection at the individual-flow level, future work could incorporate detection methods at the device level, while acknowledging the potential limitations imposed by the restricted computational resources of IoT devices.
Second, it could explore coordinated multivector attacks that simultaneously target both the satellite network and the IoT server beyond the backhaul. 
Third, advanced attackers often employ multistage strategies and demonstrate high adaptability, underscoring the need for dynamic game models and evolving defense mechanisms to effectively counter such threats.
Finally, regarding malicious packet detection, real-world attack scenarios can involve implicit or obfuscated features. Addressing this challenge will require the integration of machine learning and AI techniques to enhance detection accuracy and reduce false alarm rates.
}

\appendices

\section{Proofs of Lemmas \ref{lemma:how alpha3 affects b}, \ref{lemma:Attacker's optimal malicious flow}, and \ref{lemma:monotonicity}}
\label{app:proof of Lemma 3,4,5}
    {The proof of Lemma \ref{lemma:how alpha3 affects b} is presented as follows.}  
    According to \eqref{eq:closed form of b}, $1/b^{Idl}=1+{\alpha^{SeR}}/({\alpha^{AT}} + {\alpha^{SyR}} + {\alpha^{SeT}})+({\alpha^{SeR}}{\alpha^{AT}})/[{\alpha^{SyR}}({\alpha^{AT}} + {\alpha^{SyR}} + {\alpha^{SeT}})]$, which decreases concerning ${\alpha^{SyR}}\in [0,1]$. 
    Similarly, $b^{Nor} ={\alpha^{SeR}}/({\alpha^{SeR}}+{\alpha^{AT}}+{\alpha^{SeT}}+{\alpha^{SyR}}+({\alpha^{SeR}}{\alpha^{AT}})/{\alpha^{SyR}})$. 
    A direct derivative analysis shows that ${\alpha^{SyR}}+({\alpha^{SeR}}{\alpha^{AT}})/{\alpha^{SyR}}$ increases and decreases concerning ${\alpha^{SyR}}\in [0,\sqrt{{\alpha^{SeR}} {\alpha^{AT}}}]$ and ${\alpha^{SyR}}\in [\sqrt{{\alpha^{SeR}} {\alpha^{AT}}},1]$. 

{The proof of Lemma \ref{lemma:Attacker's optimal malicious flow} is presented as follows.}  
    Taking the first-order derivative of $U_n$ over $\mu_n^{Mal}\in\mathbb{R}^{0+}$ {and setting it equal to zero yields} \eqref{eq:Attacker's optimal malicious flow}. Since $m\geq 2$ and $\mu_n^{Nor}\in\mathbb{R}^{0+}$, we can verify that the second-order derivative {is negative, indicating that Equation \eqref{eq:Attacker's optimal malicious flow} represents a maximum.}  

{The proofs of Lemma \ref{lemma:monotonicity} is presented as follows.}  
The derivative of $\frac{(m-1)^{m-1}}{m^m}$ is $\frac{(m-1)( \ln((m-1)/m)}{m^m}$, which is smaller than $0$ for all $m\in\mathcal{M}$. Thus, $\hat{w}^{Mal}$ decreases with $m\in \mathcal{M}$. As an increasing function of $\hat{w}^{Mal}$, $\hat{P}^{Cap}$ decreases with $m\in \mathcal{M}$ {for each $a\in\mathcal{A}$}. 
Since $-1/m$ increases with $m\in \mathcal{M}$, both $(1-1/m)^m$ and $\hat{w}^{Nor}$ increase with $m\in \mathcal{M}$.

\section{{The proofs of Corollaries \ref{corollary:1 and K} and \ref{corollary:Fundamental Limit of Security-Efficiency Ratio}}}
\label{app:2corollary}

{The proofs of Corollary \ref{corollary:1 and K} is presented as follows.}  
    If $\tilde{m}\in \mathcal{M}^{>}$, then $(\rho(\hat{a}^*(\tilde{m}))-\rho(a_i))g^{SP}(\tilde{m})\hat{P}^{Cap}(\tilde{m})\geq -(c^{OC}(\hat{a}^*(\tilde{m}))-c^{OC}(a_i)), \forall a_i\in \mathcal{A}$, based on statement (b) of Proposition \ref{proposition:two regions}. 
    If $m\in [m_1,\tilde{m}]$, then we obtain $m\in \mathcal{M}^{>}$ based on Lemma \ref{lemma:Threshold Set-Division}. 
    Since $g(m)>0$, the product of two positive decreasing functions, i.e., $g^{SP}(m)\hat{P}^{Cap}(m)$, decreases as $m\in \mathcal{M}$ increases. Thus, $g^{SP}(m)\hat{P}^{Cap}(m)\geq g^{SP}(\tilde{m})\hat{P}^{Cap}(\tilde{m})>0$. 
    Since $\rho(\hat{a}^*(\tilde{m}))\geq \rho(a_i), \forall a_i\in \mathcal{A}$, we obtain $(\rho(\hat{a}^*(\tilde{m}))-\rho(a_i))g^{SP}(m)\hat{P}^{Cap}(m)\geq (\rho(\hat{a}^*(\tilde{m}))-\rho(a_i))g^{SP}(\tilde{m})\hat{P}^{Cap}(\tilde{m}) \geq -(c^{OC}(\hat{a}^*(\tilde{m}))-c^{OC}(a_i)), \forall a_i\in \mathcal{A}$. 

{The proofs of Corollary \ref{corollary:Fundamental Limit of Security-Efficiency Ratio} is presented as follows.} 
Since $g^{SP}(m)$ decreases in $m$, $g^{SP}(m)\leq g^{SP}(2)=(r^{Cap}-r^{SP})h^{Mal}/4-r^{SP}h^{Nor}, \forall m\in \mathcal{M}$. If $r^{Cap}/r^{SP}\leq 2+4{\alpha^{SyR}}/{\alpha^{AT}}$, then $(r^{Cap}-r^{SP})h^{Mal}/4-r^{SP}h^{Nor}\leq 0$ and $g^{SP}(m)\leq 0$ for all $m\in \mathcal{M}$. 
Then, the above corollary follows from the statement (a) of Proposition \ref{proposition:two regions}. 



%



\bibliographystyle{IEEEtran}
\bibliography{SIsecurity}

\begin{IEEEbiography}[{\includegraphics[width=1in,height=1.25in,clip,keepaspectratio]{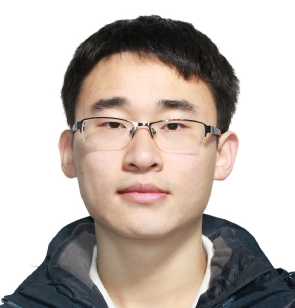}}]{Linan Huang} 
received the B.Eng. degree in Electrical Engineering from Beijing Institute of Technology, China, in 2016 and the Ph.D. degree in electrical engineering from New York University (NYU), Brooklyn, NY, USA, in 2022. 
He is currently an assistant researcher at Tsinghua University. 
His research interests include dynamic decision-making in the multi-agent system, mechanism design, artificial intelligence, cybersecurity, and satellite networks. 
\end{IEEEbiography}

\begin{IEEEbiography}[{\includegraphics[width=1in,height=1.25in,clip,keepaspectratio]{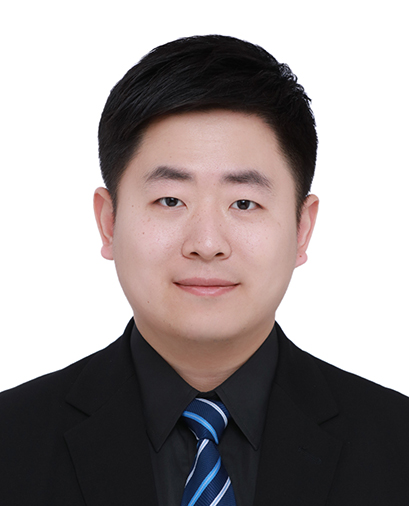}}]{Peilong Liu} 
received his B.S. degree from Tianjin University, Tianjin, China, in 2012, and the Ph.D. degree from University of Chinese Academy of Sciences, Beijing, China, in 2018. He is currently an assistant researcher with the Beijing National Research Center for Information Science and Technology, Tsinghua University. His research interests include satellite network traffic engineering and satellite communications. 
\end{IEEEbiography}

\begin{IEEEbiography}[{\includegraphics[width=1in,height=1.25in,clip,keepaspectratio]{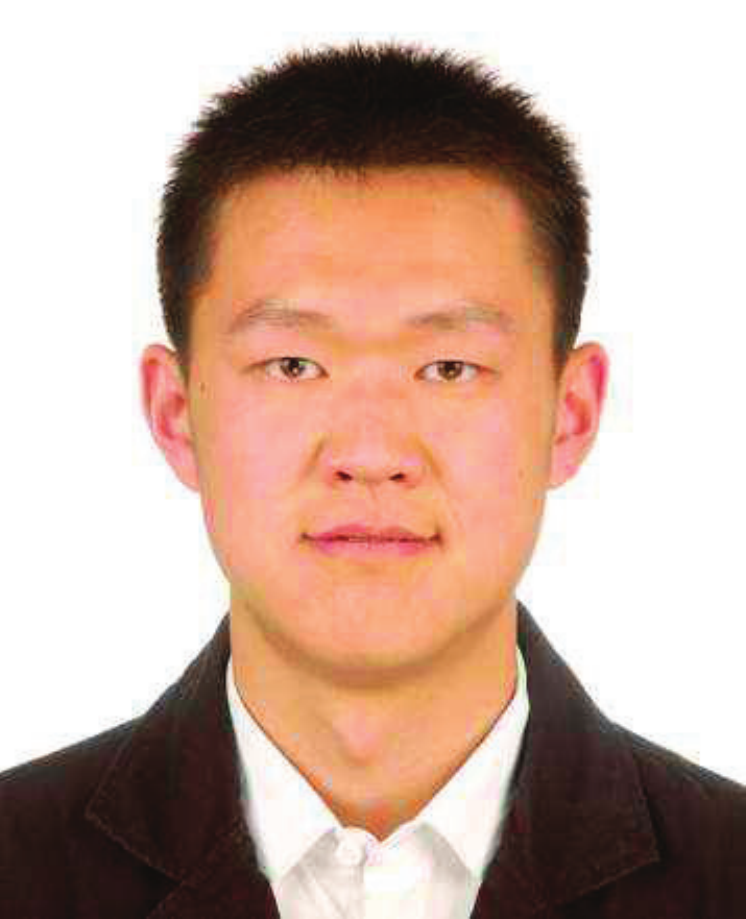}}]{Xu Chen} 
received his B.S. and M.S. degrees in School of Mathematical Sciences from Peking University in 2009 and 2013, respectively, and the Ph.D. degree in the Department of Electronic Engineering from Tsinghua University in 2022. He has worked as an assistant research fellow for several years. He is now a postdoctoral researcher at Tsinghua University. His research interests include game theory, network management and security.
\end{IEEEbiography}

\begin{IEEEbiography}[{\includegraphics[width=1in,height=1.25in,clip,keepaspectratio]{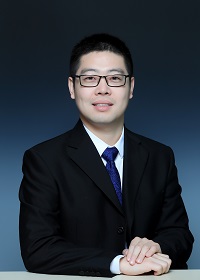}}]{Chunxiao Jiang} 
received the B.S. degree in information engineering from Beihang University, China, in 2008, and the Ph.D. degree in electronic engineering from Tsinghua University, China, in 2013. 
He is currently an Associate Professor at the Tsinghua Space Center, Tsinghua University. His research interests include the application of game theory, optimization, and statistical theories to communication, networking, and resource allocation problems, in particular heterogeneous space networks. 
\end{IEEEbiography}

\begin{IEEEbiography}[{\includegraphics[width=1in,height=1.25in,clip,keepaspectratio]{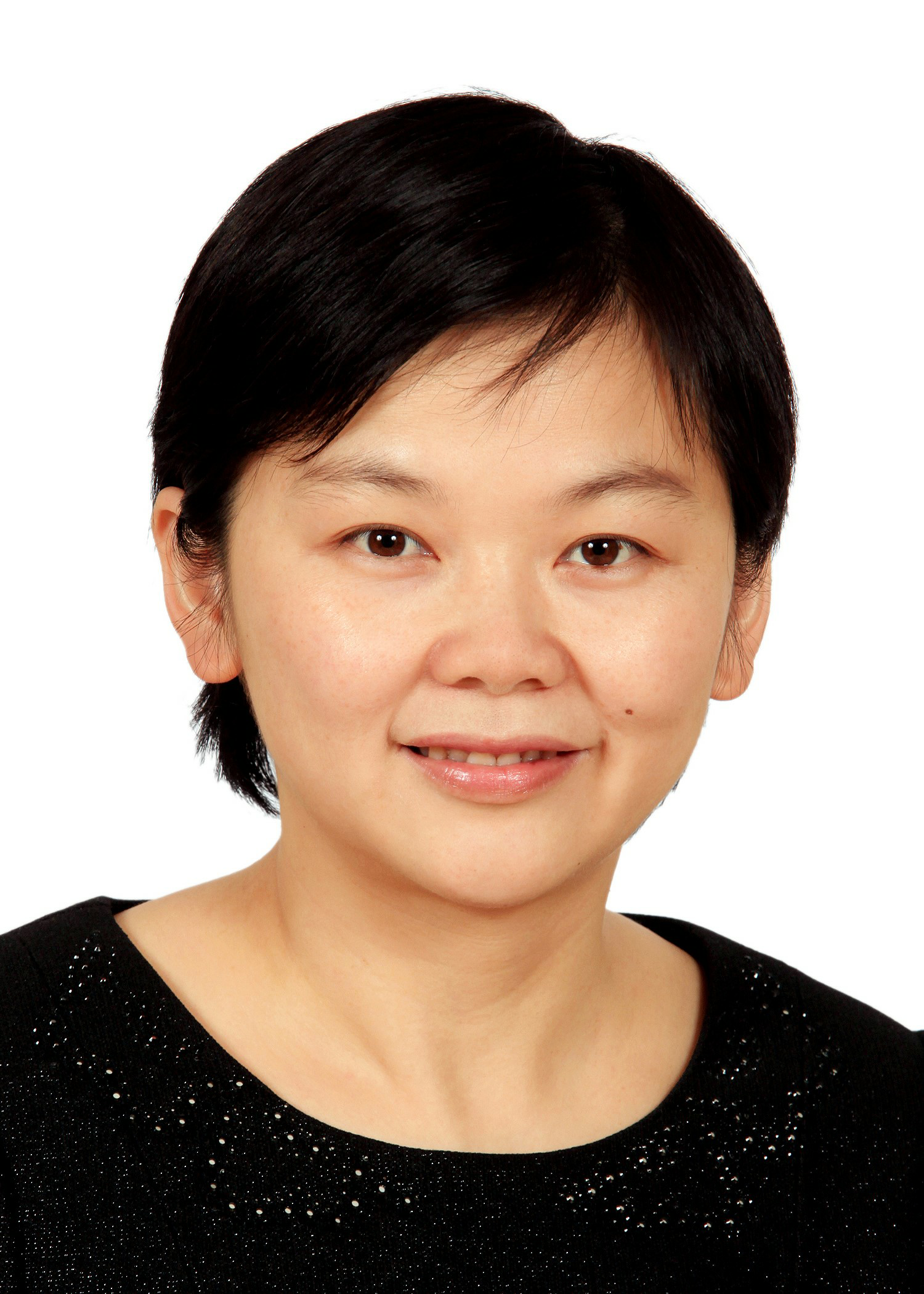}}]{Linling Kuang}  received the B.S. and M.S. degrees from the National University of Defense Technology, Changsha, China, in 1995 and 1998, respectively, and the Ph.D. degree in electronic engineering from Tsinghua University, Beijing, China, in 2004. Since 2007, she has been with the Tsinghua Space Center, Tsinghua University. Her research interests include wireless broadband communications, signal processing, and satellite communications. She is a member of the IEEE Communications Society.
\end{IEEEbiography}

\begin{IEEEbiography}[{\includegraphics[width=1in,height=1.25in,clip,keepaspectratio]{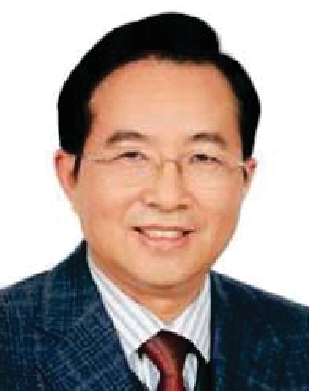}}]{Jianhua Lu} 
received his B.S. and M.S. degrees from Tsinghua University, Beijing, China, in 1986 and 1989, respectively, and his Ph.D. degree in Electrical \& Electronic Engineering from the Hong Kong University of Science \& Technology, Hong Kong, China, in 1998. 
Since 1989, he has been with the Department of Electronic Engineering, Tsinghua University.
His research interests include broadband wireless communications, multimedia signal processing, and satellite communications. 
He is an academician of Chinese Academy of Science and a Fellow of IEEE. 
\end{IEEEbiography}

\end{document}